\documentclass[sigconf]{aamas}

\usepackage{balance} 

\usepackage{tabularx}
\usepackage{array}
\usepackage{comment}
\usepackage[misc]{ifsym}
\usepackage{booktabs}
\usepackage{algorithm, algpseudocode}
\usepackage{wrapfig}
\usepackage{subcaption}
\usepackage{color}
\usepackage{url}

\newcommand{\argmin}{\mathop{\rm arg~min}\limits}

\title[Standby-Base Deadlock Avoidance Method]{Standby-Based Deadlock Avoidance Method for Multi-Agent Pickup and Delivery Tasks}

\author{Tomoki~Yamauchi}
\affiliation{
  \institution{Waseda University}
  \city{Tokyo}
  \country{Japan}}
\email{t.yamauchi@isl.cs.waseda.ac.jp}

\author{Yuki~Miyashita}
\affiliation{
  \institution{Waseda University}
  \city{Tokyo}
  \country{Japan}}
\email{y.miyashita@isl.cs.waseda.ac.jp}

\author{Toshiharu~Sugawara}
\affiliation{
  \institution{Waseda University}
  \city{Tokyo}
  \country{Japan}}
\email{sugawara@isl.cs.waseda.ac.jp}

\begin{abstract}
The {\em multi-agent pickup and delivery} (MAPD) problem, in which multiple agents iteratively carry materials without collisions, has received significant attention. However, many conventional MAPD algorithms assume a specifically designed grid-like environment, such as an automated warehouse. Therefore, they have many pickup and delivery locations where agents can stay for a lengthy period, as well as plentiful detours to avoid collisions owing to the freedom of movement in a grid. By contrast, because a maze-like environment such as a search-and-rescue or construction site has fewer pickup/delivery locations and their numbers may be unbalanced, many agents concentrate on such locations resulting in inefficient operations, often becoming stuck or deadlocked. Thus, to improve the transportation efficiency even in a maze-like restricted environment, we propose a deadlock avoidance method, called {\em standby-based deadlock avoidance} (SBDA). SBDA uses {\em standby nodes} determined in real-time using the {\em articulation-point-finding algorithm}, and the agent is guaranteed to stay there for a finite amount of time. We demonstrated that our proposed method outperforms a conventional approach. We also analyzed how the parameters used for selecting standby nodes affect the performance.
\end{abstract}


\def\tsk{{\it tsk}}
\def\TskEPSet{V_\tsk}
\def\dist{\textrm{dist}}
\def\material{\phi}

\def\wait{{\it wait}}
\def\move{{\it move}}
\def\load{{\it load}}
\def\unload{{\it unload}}
\def\rotate{{\it rotate}}
\def\TaskSet{\mathcal{T}}
\def\park{{\it park}}

\def\AgentSet{{A}}

\def\StandByNodeSet{\mathcal{S}}
\def\PStandByNodeSet{\StandByNodeSet_{\it psn}}

\def\dist{{\it dist}}

\def\psn{{\it psn}}
\def\CrowdedList{CL}

\begin{document}


\pagestyle{fancy}
\fancyhead{}


\maketitle


\section{Introduction}
The use of {\em multi-agent systems} (MAS) for complex and enormous tasks in real-world applications has recently attracted considerable attention. Examples include transport robots in automated warehouses~\cite{wurman2008coordinating}, autonomous aircraft-towing vehicles~\cite{morris2016planning}, ride-sharing services~\cite{li2019efficient,yoshida2020multi}, office robots~\cite{veloso2015cobots}, and delivery systems with multiple drones~\cite{krakowczyk2018developing}. However, simply increasing the number of agents may lead to inefficiency owing to redundant movements and resource conflicts such as collisions. Therefore, to improve the overall performance, it is essential to have coordinated actions that avoid negative effects among agents. In particular, collision avoidance is essential in our application, which is a pickup and delivery system applied in a restricted environment where the agents are large heavy-duty robots carrying heavy and large-sized materials.
\par

This problem is formulated as a {\em multi-agent pickup and delivery} (MAPD) problem, where pickup and delivery tasks are assigned to individual agents simultaneously. The assigned agent then needs to move to the material storage area, load a specified material, deliver it to the required location, and unload it. Therefore, because there are numerous tasks to fulfill, the MAPD problem can be considered as an iteration of {\em multi-agent path-finding} (MAPF), in which multiple agents generate collision-free paths to their targets without collisions. Unfortunately, the MAPF problem is known to be an NP-hard problem to obtain the optimal solution~\cite{ma2016multi}, and thus the MAPD problem is more time-consuming.
\par

There have been many studies focusing on MAPF/MAPD problems~\cite{sharon2015conflict,ma2019lifelong,liu2019task,ma2017lifelong,nie2020effective,li2021lifelong,okumura2019priority}, and their results have been used in real-world applications. With these applications, generating plans without collisions and deadlocks becomes the central issue. For example, Okumura et al.~\cite{okumura2019priority} proposed {\em priority inheritance with backtracking} (PIBT) in which the agents decide the actions within a short time window to avoid deadlocks based on priorities and through local communication. Ma et al.~\cite{ma2017lifelong} proposed {\em token passing} (TP) with using the {\em holding task endpoints} (HTE)~\cite{liu2019task,li2021lifelong}, which exclusively holds the {\em endpoints} including the pickup/delivery nodes of the tasks, to avoid collisions. Liu et al.~\cite{liu2019task} introduced {\em reserving dummy paths} (RDP), which consistently reserve {\em dummy paths} to a unique parking location of each agent, allowing the agents to fulfill their tasks with the same endpoint simultaneously.
\par

However, these methods cannot be used in our target environments. For example, to guarantee completeness, PIBT requires the environment to be bi-connected; however, our environment cannot meet this requirement. For efficient movement, HTE and RDP assume a grid-like environment, which has many endpoints where agents can stay for any time length and/or many detours whose lengths are almost identical. Such requirements are possible in specially designed environments such as an automated warehouse. However, our environments are {\it ad hoc} and maze-like, such as a search-and-rescue or construction site, and usually have fewer pickup and/or delivery locations whose numbers may be unbalanced, resulting in congestion near a few of the endpoints. Moreover, detours are limited and often become much longer.
\par

Therefore, we propose a deadlock avoidance method, i.e., {\em standby-based deadlock avoidance} (SBDA), for improving the transportation efficiency while allowing agents to conduct their tasks with the same endpoint even in a maze-like environment. We integrate this method with TP although the number of endpoints is quite small in our target environments. The SBDA algorithm employs {\em standby nodes}, where an agent is guaranteed to wait for a finite amount of time near its destination (i.e., its endpoint). The set of standby nodes changes while the agents reserve such nodes to stay there; however, they can be determined efficiently in real-time using the {\em articulation-point-finding algorithm} (APF algorithm) in the graph theory~\cite{tarjan1972depth}. When an agent moves toward the endpoint where other agents have already arrived and/or are traveling as their destination, SBDA enables the agent to temporarily wait at one of standby nodes remaining near the endpoint to avoid becoming stuck or entering a deadlock and head toward its destination in turn. SBDA is the suboptimal algorithm in terms of transportation efficiency but guarantees completeness by using standby nodes. Herein, we evaluate the performance of our proposed method and compare it with that using HTE as a baseline under various experimental conditions. We then demonstrate that our proposed method outperforms the baseline method in maze-like restricted environments. Finally, we analyze the features of the proposed method by conducting experiments under various parameter settings and applying an ablation study.


\section{Related Work}
Studies on the MAPF/MAPD problem have been approached from a variety of perspectives~\cite{felner2017search,ma2017overview,salzman2020research}. One main approach is centralized planning and scheduling~\cite{sharon2015conflict,bellusci2020multi,boyarski2015icbs,boyrasky2015don,liu2019task}. For example, Sharon et al.~\cite{sharon2015conflict} proposed the {\em conflict-based search} (CBS) algorithm for the MAPF problem. CBS is a two-stage search approach, consisting of a {\em low-level search} in which agents generate their paths independently, and a {\em high-level search} in which a centralized planner generates optimal collision-free paths by receiving local plans from all agents. CBS and its extensions are used in many other centralized planners~\cite{bellusci2020multi,boyarski2015icbs,boyrasky2015don,liu2019task}. Although we can expect optimal solutions/paths for the MAPD instances, the computational cost of a centralized planner rapidly increases with an increase in the number of agents~\cite{ma2017lifelong}.
\par

Decentralized methods are more scalable and robust, and thus many studies have been conducted from this perspective~\cite{ma2017lifelong,okumura2019priority,wang2020walk,wang2011mapp}. However, because the plans are generated individually, such methods for the MAPD problem must have completeness as well as the ability to detect and resolve conflicts between them under certain restrictions. For example, Okumura et al.~\cite{okumura2019priority} proposed the PIBT for the MAPF/MAPD problem, in which agents determine the next nodes in accord with their priorities through communication with local agents. The PIBT is effective and has been extended for use in more general cases~\cite{okumura2019winpibt,okumura2021time}; however, the completeness is guaranteed only in bi-connected environments. Ma et al.~\cite{ma2017lifelong} proposed a TP in which an agent selects tasks whose endpoints have not yet been reserved by other agents, and generates its collision-free plan by exclusively referring to the token, i.e., a synchronized shared block of memory. However, many of these studies also assume that the environments are grid-like such that, to avoid collisions, there are many escape nodes and the same-length of detours between two points. However, applying them to our maze-like environments may lead to a reduced efficiency of transportation and planning owing to the small number of endpoints and few detours, whose lengths are quite different.
\par

By contrast, some studies~\cite{kala2012multi,damani2021primal,huang2021learning,ren2021ms,ren2021multi,zhang2020multi} have also considered an application to maze-like environments. For example, Damani et al.~\cite{damani2021primal} proposed {\em pathfinding via reinforcement and imitation multi-agent learning - lifelong} (PRIMAL$_2$), a distributed reinforcement learning framework for a {\em lifelong} MAPF (LMAPF), which is a variant of the MAPF in which agents are repeatedly assigned new destinations. However, they assumed that tasks are sparsely generated at random locations, and thus, unlike our environment, no local congestion occurs. Although other studies in the area of trajectory planning and robotics also assume maze-like environments~\cite{robinson2018efficient,tahir2019heuristic,dewang2018robust}, our study differs because we aim at completing the MAPD instance efficiently in a restrictive maze-like environment.
\par

Previous studies have focused on deadlock avoidance in the MAPD problem~\cite{liu2019task,ma2017lifelong,nie2020effective,li2021lifelong}, similar to our approach. For example, HTE~\cite{liu2019task,li2021lifelong} method assumes that the environment has many endpoints where agents can remain for a finite length of time; otherwise, the performance will decrease because fewer tasks can be executed in parallel. By contrast, RDP~\cite{liu2019task} always reserves dummy paths to the unique parking location of each agent to allow the agents to conduct their tasks with the same endpoint simultaneously; hence, numerous detours whose lengths are almost identical are necessary for achieving efficiency. However, our environment is maze-like and has fewer endpoints, and thus the use of endpoints is limited~\cite{nie2020effective}. It also has a limited number of detours whose lengths may differ considerably.
\par


\section{Preliminaries}
\subsection{Problem Formulation}
The MAPD problem consists of an agent set $\AgentSet=\{1, \dots,M\}$, a task set $\TaskSet=\{\tau_1, \dots,\tau_N\}$, and an undirected connected graph $G = (V,E)$ embeddable in a two-dimensional Euclidean space. Node $v \in V$ corresponds to a location, and an edge $(u,v) \in E$ ($u, v\in V$) corresponds to a path along which an agent can move between locations $u$ and $v$. We can naturally define the {\em length} of edge $(u,v)$ by denoting $l(u,v)$. The {\em distance} between nodes $v_1$ and $v_2$ is defined as the sum of the lengths of edges appearing in the shortest path from $v_1$ to $v_2$. We assume that an endpoint is set to a dead-end that has only one associated edge in $G$. Our agent is a heavy-duty forklift-like autonomous robot with a picker in front, which carries a heavy material (500 kg to 1 ton) and can pick up (load) or put down (unload) this material using a picker in a specific direction at a specific node. We introduce discrete-time $t \in \mathbb{Z}^+$, where $\mathbb{Z}^+$ is the set of positive integers.
\par

For agent $i\in \AgentSet$, we define the {\em orientation} $o_i^t \in \mathbb{Z}^+$ and moving {\em direction} $d_i^t \in \mathbb{Z}^+$ of $i$ at time $t$, where $0 \leq o_i^t, d_i^t < 360$ in $D$ increments, and $o_i^t = 0$ and $d_i^t = 0$ indicate the northward orientation and direction of $G$. In addition, the set of possible orientations is denoted as $\mathcal{D}$. For example, if $D = 90$, $\mathcal{D}=\{0,90,180,270\}\ni o_i^t,d_i^t$. We assume $D = 90$ for simplicity, whereas $D$ can have any number depending on the environmental structure.
\par

Agents can conduct the following actions: $\move$, $\rotate$, $\wait$, $\load$, and $\unload$ on any node. Using the length $l = l(u,v)$, the rotation angle $\theta \in \mathbb{Z}^+$, and the waiting time $t$, the durations of actions $\move$, $\rotate$, $\wait$, $\load$, and $\unload$ are denoted by $T_{mo}(l)$, $T_{ro}(\theta)$, $T_{wa}(t)$, $T_{ld}$, and $T_{ul}$, respectively. Suppose that $i$ is on $v$ at time $t$. By action $\move$, $i$ moves forward or backward along edge $(u,v)$ to $u$ after appropriately changing its orientation $o_i^t$ through a rotation. By $\rotate$, $i$ rotates $D$ degrees clockwise $(D)$ or counter-clockwise $(-D)$ from $o_i^t$, i.e., $o_i^{t+T_{ro}(D)} = o_i^t \pm D$, at $v$. Agent $i$ has a unique parking node $\park_i \in V$~\cite{liu2019task}, which is the starting location at $t = 0$, and returns and remains there if $i$ has no tasks to execute. Parking nodes are expressed by the red squares in Fig.~\ref{fig:environment}.
\par

The task $\tau_j$ is specified by the tuple $\tau_j=(\sigma_{\tau_j}^{ld},\sigma_{\tau_j}^{ul},\material_{\tau_j})$, where $\sigma_{\tau_j}^{ld}=(v_{\tau_j}^{ld},o_{\tau_j}^{ld})$ ($\in V\times\mathcal{D}$) are the location and orientation when loading a material $\material_{\tau_j}$, and $\sigma_{\tau_j}^{ul}=(v_{\tau_j}^{ul},o_{\tau_j}^{ul})$ ($\in V\times\mathcal{D}$) are the location and orientation when unloading a material $\material_{\tau_j}$. When an agent loads and unloads a material, it needs to be oriented in a specific direction, considering the direction of the picker. Agents have to complete all tasks in $\TaskSet$ without collisions or deadlocks and then return to their own $\park_i$.
\par

\subsection{Token Passing} \label{sec:TokenPassing}
Although the proposed SBDA can be adapted to some existing MAPD algorithms, in this paper, we show that it can be integrated into TP to efficiently solve an MAPD instance. TP~\cite{ma2017lifelong} is a well-known MAPD algorithm, in which agents choose tasks themselves and generate paths using the information in the token. A token is a synchronized shared memory block containing the current paths of all agents, tasks currently assigned agents, and the remaining tasks that are not assigned agents.
\par

The set of endpoints $V_{\it ep}$ ($\subset V$) of an MAPD problem consists of all {\em task endpoints} that are possible pickup/delivery locations of the tasks and {\em non-task endpoints} including the initial location (i.e., $\park_i$) of each agent.
TP assumes that any agent can stay at an endpoint for any finite
length of time without blocking the movement of other agents for their
current tasks. We denote the set of task
endpoints by $\TskEPSet$ ($\subset V_{\it ep}$). It is assumed that
the locations of all endpoints are given to the agents in
advance. Hence, when an agent chooses one task whose pickup/delivery
locations are {\em open} endpoints, meaning that they do not appear as
the endpoints of the executing plans in the current token, a
collision-free path is generated by looking at the content of the
token. HTE in TP then holds the task endpoints to avoid
collisions.
\par

Although not all MAPD instances are solvable, {\em well-formed} MAPD instances are always solvable~\cite{vcap2015complete}. For TP, an MAPD instance is well-formed if and only if (a) the number of tasks is finite, (b) there are not fewer non-task endpoints than agents, and (c) a path exists between any two endpoints that does not traverse other endpoints~\cite{ma2017lifelong}. In well-formed MAPD instances, agents can move to non-task endpoints (e.g., $\park_i$) at any time and stay there as long as necessary to avoid collisions with other agents. This action might be able to reduce the number of agents to avoid an overly crowded environment. Obviously, the assumption regarding the endpoints makes these requirements hold. However, agents cannot conduct tasks simultaneously if their endpoints are overlapped; hence, this approach reduces the efficiency considerably in maze-like environments such as our considered environment.
\par


\section{Standby-Based Deadlock Avoidance}
\subsection{Status Management Token}
We propose a novel deadlock avoidance method, i.e., SBDA, to improve the transportation efficiency while allowing agents to conduct tasks having the same endpoints even in a maze-like environment. We introduce the {\em status management token} (SMT), which is an extension of a token~\cite{ma2017lifelong} and a {\em synchronized block of information}~\cite{yamauchi2021path}, to manage the status of the planning, agents, and standby nodes and to detect conflicts with other agents, where we define a conflict as a situation in which multiple agents occupy the same node $v\in V$ or cross the same edge simultaneously. SMT contains the {\em reservation table} (RT), the {\em task execution status table} (TEST), and the {\em standby-node status table} (SST). RT is the set of RT tuples $(v, [s_v^i, e_v^i], i)$, which are valid reservation data used to generate collision-free plans, where $[s_v^i, e_v^i]$ is the occupancy intervals of agent $i$ for node $v$. When $e_v^i < t_c$ ($t_c$ is the current time), the tuple has been expired and deleted from RT. TEST is the set of TEST tuples $(\tau, v, i)$, where $\tau$ is the task currently being executed by $i$, and $v$ is the associated destination, which is the load or unload node specified in $\tau$. Therefore, two TEST tuples $(\tau, v_\tau^{ld}, i)$ and $(\tau, v_\tau^{ul}, i)$ are added when $i$ selects task $\tau$, and when $i$ arrived at $v_\tau^{ld}$ or $v_\tau^{ul}$, the corresponding entry is removed from TEST. SST is used to manage the status of all standby nodes, which are dynamically referred to and modified by SBDA. The structure of SST is described in Section~\ref{sec:FondingPSN}.
\par

SMT is a sharable memory area, and similar to a token in TP, only one agent can exclusively access it at a particular time. Although SMT may incur a slight performance bottleneck, the movement speeds of the robots are not fast and thus the time required for an overhead owing to a mutual exclusion is negligible for a realistic number of agents (e.g., less than 30 agents in the experiment environments as shown in Fig.~\ref{fig:environment}).

\subsection{Finding Potential Standby Nodes}\label{sec:FondingPSN}
Standby nodes are intuitively nodes by which, similar to non-task endpoints in TP, an agent is guaranteed to wait for a finite length of time on its way to the endpoint of the assigned task; however, they can be identified dynamically. We assume that, except for parking nodes, non-task endpoints are not given in advance. Let us define an {\em articulation point} (AP) and a {\em potential standby node}, which will be used as a standby node when needed:
\begin{definition}\label{def:ap}
A node in graph $G$ is an {\em articulation point} iff removing it and the associated edges will split the connected area of $G$ and increase the number of connected components.
\end{definition}
\begin{definition}\label{def:psn}
A {\em potential standby node} in graph $G$ is a node that is neither an AP, an endpoint, nor a dead-end in $G$.
\end{definition}
Let us denote the set of all potential standby nodes in $G = (V,E)$ by $\PStandByNodeSet(G)$ ($\subset V$). The next proposition is obvious from this definition.
\begin{proposition}\label{prop:connected}
If $G = (V,E)$ is a connected graph and $\PStandByNodeSet(G) \not= \varnothing$, then the subgraph generated by eliminating any node $v_\psn$ in $\PStandByNodeSet(G)$ from $G$ is connected.
\end{proposition}
Therefore, even if an agent remains at a potential standby node, other agents can generate paths to reach their destinations without passing that node. We therefore have the following:
\begin{corollary}\label{cor:staying}
Any agent reserving a standby node can remain there for a finite length of time without blocking the movements of other agents.
\end{corollary}
The set of all potential standby nodes $\PStandByNodeSet(G)$ is efficiently identified using the APF algorithm such as Tarjan's algorithm~\cite{tarjan1972depth}, whose time complexity is O(|V|+|E|)~\cite{NUUTILA19949}. Now, we define a {\em standby node}, where $G_t$ is the modified subgraph of $G$ at time $t$, as below.
\begin{definition}
The potential standby node $v\in\PStandByNodeSet(G_t)$ is a {\em standby node} when agent $i\in\AgentSet$ reserves $v$ to remain there.
\end{definition}
\noindent
When the agent leaves standby node $v$, $v$ is no longer in standby node. We denote the set of all standby nodes with a reserving agent as $\StandByNodeSet_t=\{(v, i))\ |\ i\in\AgentSet \textrm{ reserves } v \textrm{ as standby node}\}$.
\par

When agents make a reservation to stay at a number of standby nodes (how such a reservation is created is explained in Section~\ref{sec:UseStandbyNodes}), other agents are prohibited to pass through them. This means that the structure of graph $G$ is temporally modified by eliminating the standby nodes (and the associated edges). This modified graph at time $t$ is then denoted by $G_t = (V_t, E_t)$. Thus, the set of potential standby nodes in $G_t$ is denoted by $\PStandByNodeSet(G_t)$, which is also efficiently identified using the APF algorithm. Note that $G = G_0$.
\par

Before agents start the tasks in an MAPD problem, SBDA generates the set of {\em associated potential standby nodes} $s(v_\tsk)$ for every endpoint $v_\tsk$. For endpoint $\forall v_\tsk$ ($\in \TskEPSet)$, SBDA calculates $s(v_\tsk)$ ($\subset \PStandByNodeSet(G)$) using
\[
s(v_\tsk) = \{v_{\it psn} \in\PStandByNodeSet(G) | \dist(v_{\it psn},v_\tsk) \leq
\alpha \},
\]
where $\dist(v_1,v_2)$ is the distance between $v_1$ and $v_2$, i.e., the length of the shortest path between them. Parameter $\alpha$ ($\geq0$) is the limit of the distance to a standby node from the endpoint $v_\tsk$. Note that it is possible that $s(v_\tsk) = \varnothing$ and $s(v_\tsk)\cap s(v_\tsk') \not= \varnothing$ for $v_\tsk,v_\tsk'\in \TskEPSet$.
Similarly, the set of potential standby nodes for $v_\tsk$ at $t$ is denoted by $s_t(v_\tsk)=s(v_\tsk)\cap\PStandByNodeSet(G_t)$. Therefore, $s(v_\tsk)=s_0(v_\tsk)$ is the set of the initial potential standby nodes for $G$ ($=G_0$). We also define the set of potential standby nodes that are not included in the associated potential standby nodes as
\[
\PStandByNodeSet^c(G_t) =
\PStandByNodeSet(G_t)\setminus\bigcup_{v\in \TskEPSet} s_t(v).
\]
An element in $\PStandByNodeSet^c(G_t)$ is called a {\em free potential standby node at $t$}. Note that it is possible that $\PStandByNodeSet^c(G_t) = \varnothing$.
\par

Information regarding the standby nodes is stored in SST. SST at $t$ consists of (1) the initial set of potential standby nodes, $\PStandByNodeSet(G)$, (2) the initial set of all pairs of endpoints and their potential standby nodes $\{(v_\tsk,s(v_\tsk))\ |\ v_\tsk\in \TskEPSet\}$, (3) the set of pairs of standby nodes and agents $i$ that reserve such nodes in the form of $(v_{sn}, i)\in\StandByNodeSet_t$, and (4) the {\em crowded list}, $\CrowdedList$ $(\subset\AgentSet)$, which is the set of agents $i$ that temporally remain at a free potential standby node in $\PStandByNodeSet^c(G)$.
\par

\begin{algorithm}[t]
 \caption{Task selection by agent $i$}
 \label{alg:taskselection}
 \begin{algorithmic}[1]
  \Function{SelectTask}{$i$}
  \State $\TaskSet'=$ Set of tasks satisfying Cond.~1.
  \State // where $\tau=(\sigma_{\tau}^{ld},\sigma_{\tau}^{ul},\material_{\tau})$,
    and $\sigma_\tau^{ld}=(v_{\tau}^{ld},o_\tau^{ld})$
    \If{$\TaskSet' \neq \varnothing$}
    \State $\tau^* \gets \argmin_{\tau \in \TaskSet'} \dist(v_c^i,v_\tau^{ld})$ // $v_c^i$: current location
    \State $\TaskSet\gets \TaskSet\setminus \tau^*$; \Return $\tau^*$
    \Else\ \Return false \label{alg:taskselectionfailure}
    \EndIf
    \EndFunction
  \end{algorithmic}
\end{algorithm}

\subsection{Task Selection Process}\label{sec:TaskSelection}
After the agents start the tasks in an MAPD problem, an agent with SBDA selects a task to conduct based on the potential standby nodes in SMT, decides its destination by selecting a standby node if necessary, and generates a path to it. The agent exclusively accesses the current SMT during this process.
\par

The pseudo-code of the task selection process {\sc SelectTask($i$)} based on the potential standby nodes is shown in Algorithm~\ref{alg:taskselection}. This is applied only when $\TaskSet\not=\varnothing$; otherwise, agent $i$ returns to its parking node $\park_i$. For $v_\psn\in \PStandByNodeSet(G_t)$, let $e_{v_\psn,t}^*$ be the last time that all agents will pass $v_\psn$ for all current plans. This can be calculated by looking at the elements $(v_\psn, [s_{v_\psn}^j, e_{v_\psn}^j], j)$ in RT of SMT, and if such an element does not exist in RT, we set $e_{v_\psn,t}^*=t_c$, where $t_c$ is the time when calculating $e_{v_\psn,t}^*$.
\par

Agent $i$ selects tasks
$\tau=(\sigma_{\tau}^{ld},\sigma_{\tau}^{ul},\material_{\tau})\in\TaskSet$ that hold the following conditions (Cond.~1) at time $t$, where $v_c^i$ is the current location of $i$, $\sigma_{\tau}^{ld}=(v_{\tau}^{ld},o_{\tau}^{ld})$ and $\sigma_{\tau}^{ul}=(v_{\tau}^{ul},o_{\tau}^{ul})$.
\begin{itemize}
\item[(1)] If $i$ remains at its parking node ($v_c^i=\park_i$), the crowded list, $\CrowdedList$, is empty.
\item[(2)] Its load node $v_\tau^{ld}$ is {\em open} (i.e., it does not appear as an endpoint in the $RT$), or $v_\tau^{ld}$ has non-empty associated potential standby nodes (i.e., $s_t(v_{\tau}^{ld})\not=\varnothing$) and $\exists v_\psn\in s_t(v_{\tau}^{ld})$ such that $e_{v_\psn}^* - t_c \leq \delta$.
\item[(3)] $|s_t(v_{\tau}^{ul})|+1$ is greater than the number of tuples that appear in $TEST$, whose destination is $v_{\tau}^{ul}$.
\end{itemize}
Here, because $i$ cannot begin to stay at $v_\psn$ before $e_{v_\psn}^*$, $\delta$ ($\geq 0$) is a threshold parameter of the margin time for reserving the standby node. The set of tasks that satisfies these conditions is denoted by $\TaskSet'$ [Line 2].
\par

Then, $i$ selects the closest task $\tau^*$, i.e., $\dist(v_c^i,v_{\tau^*}^{ld})$ in $G_t$ is the smallest, where $v_c^i$ is the current location of $i$ [Lines 4--6]. If $i$ cannot find such tasks, $i$ returns to $\park_i$. However, it is possible for $i$ to occasionally check whether other agents have completed their tasks and for some endpoints/standby nodes to become open while $i$  is going back and remains at $\park_i$.
\par

\begin{figure}[t]
  \centering
  \includegraphics[width=0.98\linewidth]{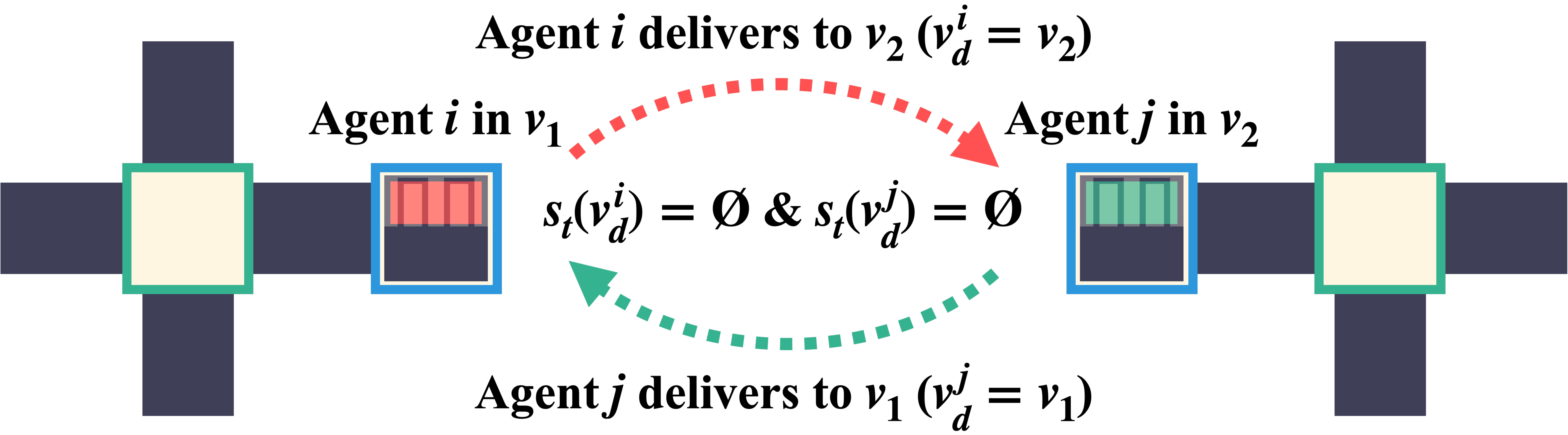}
  \caption{Deadlock state ($i$ and $j$ are on the same graph)}
  \label{fig:deadlock}
\end{figure}

\begin{algorithm}[t]
  \caption{Destination decision by agent $i$}
  \label{alg:destdecision}
  \begin{algorithmic}[1]
    \Function{DecideDest}{$i$, $v_d$, $v_c$}
    \State // $v_d$: $v_{\tau}^{ld}$, $v_{\tau}^{ul}$ or $\park_i$,
    \State // $v_{\tau}^{ld}$: pickup, $v_{\tau}^{ul}$: delivery, $\park_i$: parking node
    \If{$i\in \CrowdedList$} remove $i$ from $\CrowdedList$ \EndIf
    \If{$v_d$ satisfies one of Cond.~2, and $v_d$ is open}
    \State \Return $v_d$
    \ElsIf{$v_c\in s(v_d)$} \Return $v_c$
    \Else
    \State {$G_t^*=G_t\cup\{v_c\}$ // if $v_c\not\in\PStandByNodeSet$, $G_t=G_t^*$}
    \State {$V'_\psn(G_t^*) = \{v\in \PStandByNodeSet(G_t^*)\ |\ e_v^* - t_c \leq \delta\}$}
    \If{$S=V'_\psn(G_t^*)\cap s(v_d) \not=\varnothing$}
    \State $v_d \gets \argmin_{v\in S} (e_v^* - t_c)$
    \ElsIf{$S'=V'_\psn(G_t^*)\cap\PStandByNodeSet^c(G_t^*)\not=\varnothing$}
    \State{$v_d \gets \argmin_{v \in S'} \dist(v_d,v)$}
    \State{$\CrowdedList=\{i\}\cup \CrowdedList$}
    \Else
    \State $v_d \gets \park_i$ // $i$ returns to its parking nodes.
    \EndIf
    \EndIf
    \State \Return $v_d$
    \EndFunction
  \end{algorithmic}
\end{algorithm}

\subsection{Use of Standby Nodes}\label{sec:UseStandbyNodes}
After the previous process of task selection, agent $i$ calls the path planning process to generate a path to move to $v_\tau^{ld}$ or $v_\tau^{ul}$ of the selected task $\tau$. However, it is probable that these nodes are not open, and $i$ may instead have to generate the path to the temporary destination. Thus, before generating a path, $i$ calls function {\sc DecideDest($i$,$v_d^i$,$v_c^i$)} in Algorithm~\ref{alg:destdecision} at current time $t_c$ to decide the next actual destination, where $v_d^i$ is the destination of $i$, which is $v_\tau^{ld}$, $v_\tau^{ul}$, or $\park_i$, and $v_c^i$ is the current node of $i$, which is an endpoint, a standby node, or its parking node.
\par

Before describing Algorithm~\ref{alg:destdecision}, for $\forall i\in\AgentSet$, we introduce the set of second conditions (Cond.~2) mainly to prevent a side entry for approaching the endpoint.
\begin{itemize}
\item[(1)] $\dist(v_c^i,v_d^i)\leq\beta$, where $\beta$ ($>\alpha$) is a threshold for directly heading toward the endpoint.
\item[(2)] No agent heads to an $s_t(v_d^i)$.
\item[(3)] $v_d^i==\park_i$
\end{itemize}
If $i$ can satisfy one of these conditions, it heads directly
toward $v_d^i$ if possible; otherwise, it heads toward a standby
node of $v_d^i$. Note that $\beta$ is also the threshold preventing
other agents from being forced to wait too long at standby nodes owing
to a cutting in line.
\par

We briefly describe Algorithm~\ref{alg:destdecision}. If $v_d$
satisfies one of the conditions in Cond.~2, and $v_d$ is {\em open},
$i$ determines $v_d$ as the destination [Lines 6--7]. Note that
$\park_i$ is always open. If $v_c$ is a standby node of $v_d$, and
$v_d$ is not open, it returns $v_c$, remaining there for a while
longer [Line 8]. If not, $i$ tries to move a potential standby node of
$v_d$ as the temporary destination. Thus, it first calculates $V'_{\it
  psn}(G_t^*)$ [Line 11] by referring to RT and SST, where $G_t^*
= G_t\cup\{v_c\}$ is a subgraph of $G$ [Line 10]. If there is a
potential standby node where other agents will pass through within a
time of $\delta$, the most appropriate node is selected [Lines 12--13]
(see (2) in Cond.~1). If such a standby node does not exist, $i$
instead tries to select another free potential standby node
in $\PStandByNodeSet^c(G_t^*)$. If it can select such a node, $i$ is
added to the crowded list, $\CrowdedList$, demonstrating that it
cannot select the element of $s(v_d)$ because the environment is
crowded [Line 14--16]. Otherwise, it returns to its parking node [Line
  18]. Note that when $i$ selects and reserves the standby node as
$v_d$ [Lines 13 and 15], $(v_d, i)$ is added to $\StandByNodeSet_t$,
whereas $(v_c, i)$ is removed from $\StandByNodeSet_t$ in SST if
$i$ leaves the current standby node [Line 7].
\par

The temporal evacuation at a free standby node or parking node [Lines 14-18] allows an agent to avoid a deadlock, which occurs when agents $i$ and $j$ try to simultaneously exchange their unloading nodes and current locations and both nodes have no associated potential standby nodes. An example is shown in Fig.~\ref{fig:deadlock}, in which two agents wait until their destinations become open or $s_t(v_d)\not=\varnothing$, and this situation cannot be solved if their destinations remain unchanged. However, when one agent sets another node as a temporary destination, one of the unloading nodes become open and another agent can start to move. After $i$ arrives at standby node $v_{sn}$ instead of the endpoint specified by $\tau$, it invokes function {\sc DecideDest} to identify when $i$ can move to the actual endpoint when $i$ can obtain the right to access SMT. Then, agent $i$ generates the path to the next destination by using a path finding algorithm, releases $v_{sn}^i$, removes the corresponding entry $(v_{sn}^i,i)$ from $\StandByNodeSet_t$, unlocks SMT, and finally, leaves the standby node $v_{sn}$.
\par

Because only one agent can access SMT at a time, other agents that generate paths after agent $i$ reserves standby node $v_{sn}$ are prohibited to pass through $v_{sn}$. However, if $j$, which previously generates a path before agent $i$ reserves $v_{sn}$, can pass through it, $i$ must begin to wait at $v_{sn}$ after $j$ passes through $v_{sn}$; thus, $i$ should reach $v_{sn}$ with an appropriate delay by waiting somewhere, which may cause other delays to other agents. To reduce such a waiting time, the SBDA algorithm introduces Cond.~1 (2).
\par

Finally, we prove that our proposed method is complete for our well-formed MAPD instances. In our SBDA, we modify the well-formed condition (c) of TP described in Section~\ref{sec:TokenPassing} as follows:
\begin{itemize}
  \item[(c')] When any agent generates a plan at time $t$, there exists a path between any two nodes of the endpoints or potential standby nodes at $t$ that does not traverse other endpoints or standby nodes registered at $t$.
\end{itemize}
It is clear that the algorithm described thus far always satisfies condition (c'). In fact, the structure of graph $G$ is temporarily modified when agent $i$ selects a standby node to remain in from the potential standby nodes; however, another agent can always find a path to its destination without passing the standby nodes or endpoints in SBDA. This is because SBDA uses the APF algorithm to generate $\PStandByNodeSet(G_t)$ when $i$ selects the next task or a standby node instead of the destination in the task. Moreover, similar to TP, the agent exclusively accesses the current SMT and generates a path in turn. Therefore, we obtain the following theorem:
\begin{theorem}
SBDA solves all well-formed MAPD instances.
\end{theorem}
\begin{proof}
First, from functions {\sc SelectTask} and {\sc DecideDest}, any agent will leave from any endpoint, $v_{\it ep}\in V_{\it ep}$, immediately after a load or unload task regardless of whether it has the next task. Suppose that agent $i$ is now an endpoint or a standby node. If $i$ has no task to conduct, $i$ can generate a path to its parking node $\park_i$. Otherwise, $i$ heads to one of the endpoints, $v_{\it ep}$. If $v_{\it ep}$ is open, $i$ can generate a path to $v_{\it ep}$. If not, $i$ selects a standby node or its parking node to head toward and generates a path to the selected node. Finally, $i$ can remain at the reserved standby node or the parking node for any finite length of time (Corollary~\ref{cor:staying}), and $i$ can generate a path to $v_{\it ep}$ if $v_{\it ep}$ becomes open.
\end{proof}


\begin{table}
  \caption{Parameter values used in the experiments}\label{table:expSetting}
  \centering
  \begin{tabular}{lll}
    \toprule
    Description & Parameter & Value\\
    \midrule
    No. of agents & $M$ & 2 to 30 \\
    No. of tasks & $N$ &100 \\
    Orientation/direction increments & $D$ & 90 \\
    Duration of $move$ per length 1 & $T_{mo}(1)$ & 10 \\
    Duration of $rotate$ & $T_{ro}(D)$ & 20 \\
    Durations of $\load$ and $\unload$ & $T_{ld}, T_{ul}$ & 20 \\
    Durations of $\wait$ & $T_{wa}(t)$ & $t$ \\
    Margin time for reserving standby node & $\delta$ & 100\\
    Threshold for direct heading endpoint & $\beta$ & 20 \\
    \bottomrule
  \end{tabular}
\end{table}

\begin{figure}[t]
  \centering
  \begin{minipage}[b]{0.43\hsize}
    \centering
    \includegraphics[width=0.98\hsize]{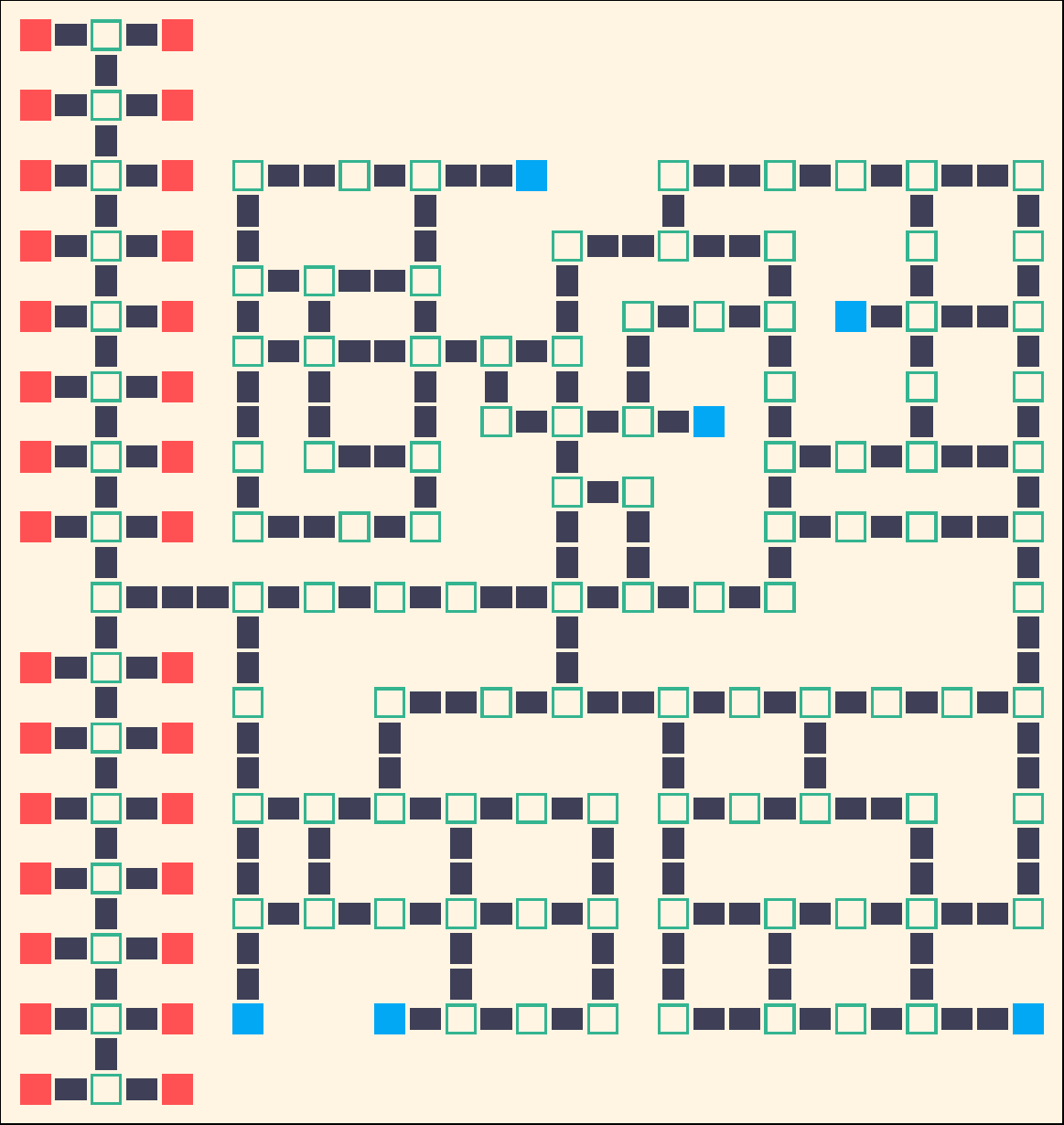}
    \subcaption{Environment 1}\label{subfig:env1}
  \end{minipage}
  \hfill
  \begin{minipage}[b]{0.52\hsize}
    \centering
    \includegraphics[width=0.98\hsize]{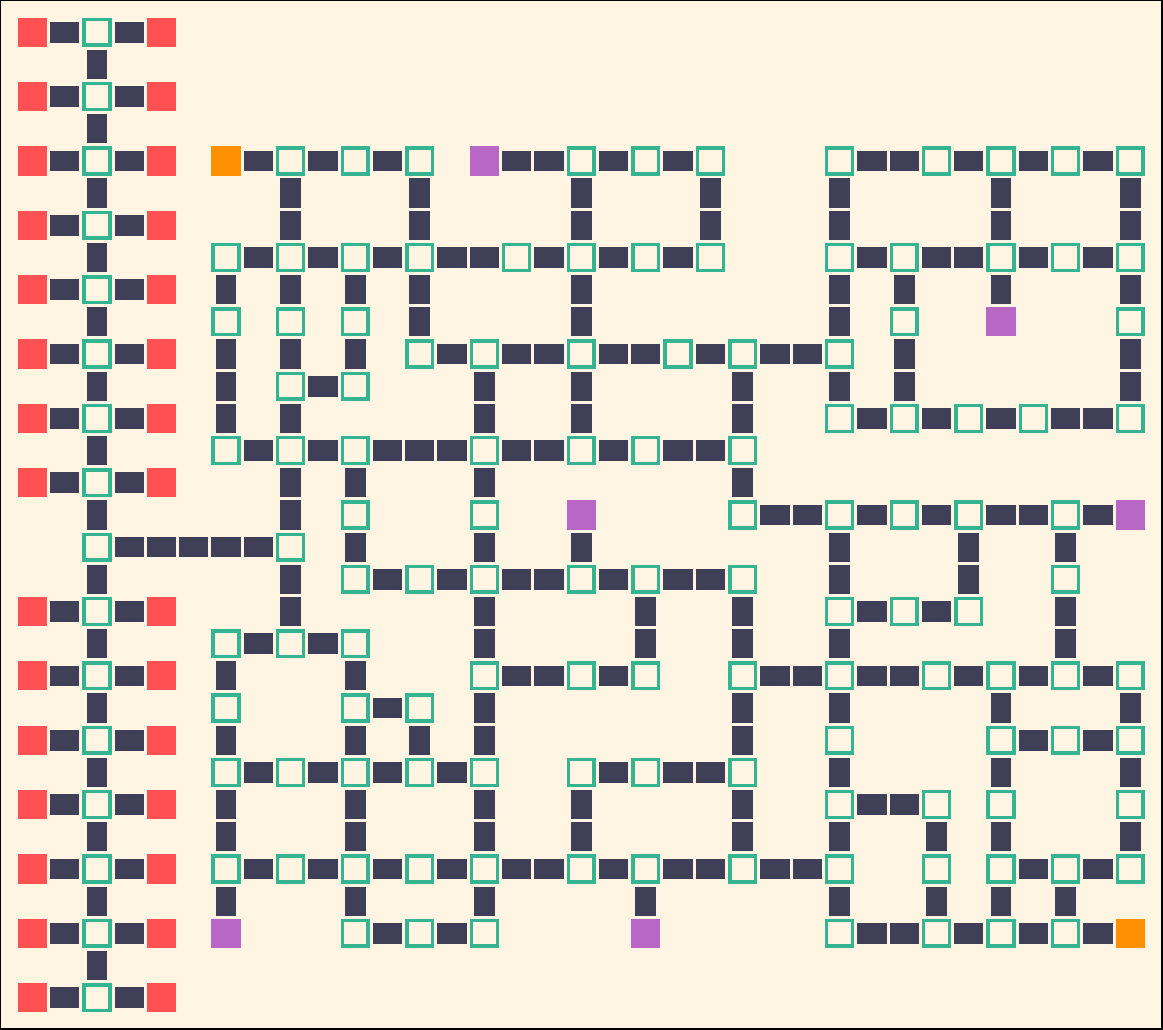}
    \subcaption{Environment 2}\label{subfig:env2}
  \end{minipage}
  \caption{Experiment environments (red, parking nodes; blue, task endpoints; orange, pickup only locations; purple, delivery only locations; hollow green, nodes; black, edges)}\label{fig:environment}
\end{figure}

\section{Experiments and Discussion} \label{sec:ExpAndDisc}
\subsection{Experiment Setting}
To evaluate the performance of our proposed method for executing the MAPD problem, we conducted the experiments under two different environments and compared the results with those using HTE as a baseline. In HTE, the agent selects the task whose loading and unloading nodes do not overlap with other endpoints of tasks that are currently being executed by other agents and whose loading node is the closest to the current node, $v_c^i$. Then, the agent generates the path to the loading node, followed by the path from the loading node to the unloading node of the task. If it cannot select such a task, it returns to the parking node. We used the space-time $A^*$ as the path finding algorithm for both methods, which is also used in TP. The heuristic function $h$ is defiend by $h(l_{ma}, \theta) = T_{mo}(l_{ma}) + T_{ro}(\theta)$, where $l_{ma}$ is the Manhattan distance between the current node and the destination, and $\theta$ is the difference between the current orientation and the required orientation at the destination. Function $h$ is clearly admissible because it does not consider the actual path length and the appropriate orientation change through a rotation before the move. The first environment (Env.~1) is a maze-like environment with few task endpoints ($|\TskEPSet| = 6$), assuming a construction site (Fig.~\ref{subfig:env1}). Nodes are set at the intersections and ends of the edges. Agents can rotate, wait, and load/unload materials only at nodes. Task endpoints are expressed by the blue squares in Fig.~\ref{subfig:env1}, where agents can both load and unload their materials. A break in the edge in this figure indicates the length of a block having a length of 1.
\par

The second environment (Env.~2) is also a maze-like environment with fewer task endpoints ($|\TskEPSet|=8$), and the number of pickup/delivery locations is skewed, as shown in Fig.~\ref{subfig:env2}. In this figure, orange squares are the pickup locations in which agents can only load their materials, and the purple squares are the delivery locations where agents can only unload their materials. The initial locations of the agents are randomly assigned to the parking nodes, which are expressed by the red squares in both environments. One hundred tasks are initially generated by randomly selecting pickup and/or delivery locations from the blue, orange, and purple squares according to the experimental setting, and added to the $\TaskSet$. Note that Envs.~1 and 2 are not clearly bi-connected.
\par

To evaluate our proposed method, we measured the {\em makespan}, i.e., the time required to complete all tasks in $\TaskSet$, and the {\em runtime}, i.e., the total CPU time for task selection, destination decision, and path planning for all tasks by all agents. Makespan indicates the transportation efficiency, whereas the runtime indicates the planning efficiency. Other parameter values are shown in Table~\ref{table:expSetting}. We conducted our experiments on a 3.00-GHz Intel 8-Core Xeon E5 with $64$ GB of RAM. The experiment results below are the average of 50 trials with different random seeds. Sample videos of our experiments can be found at \url{https://youtube.com/playlist?list=PLKufA_6vumDU01_hYNWihEN4IGlTh_oZ8}.
\par

\begin{figure}[t]
  \centering
  \begin{minipage}[b]{0.49\hsize}
    \centering
    \includegraphics[width=0.98\hsize]{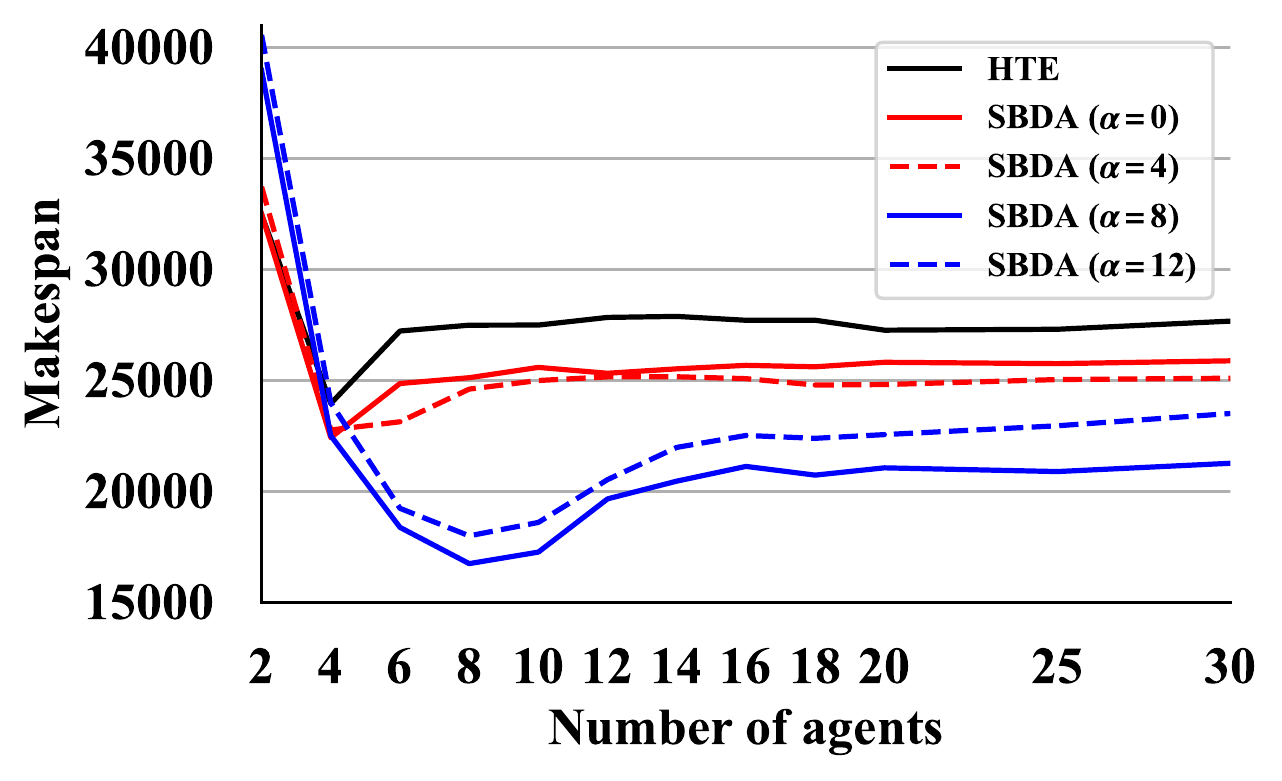}
    \subcaption{Makespan}\label{subfig:exp1-env1-makespan}
  \end{minipage}
  \hfil
  \begin{minipage}[b]{0.49\hsize}
    \centering
    \includegraphics[width=0.98\hsize]{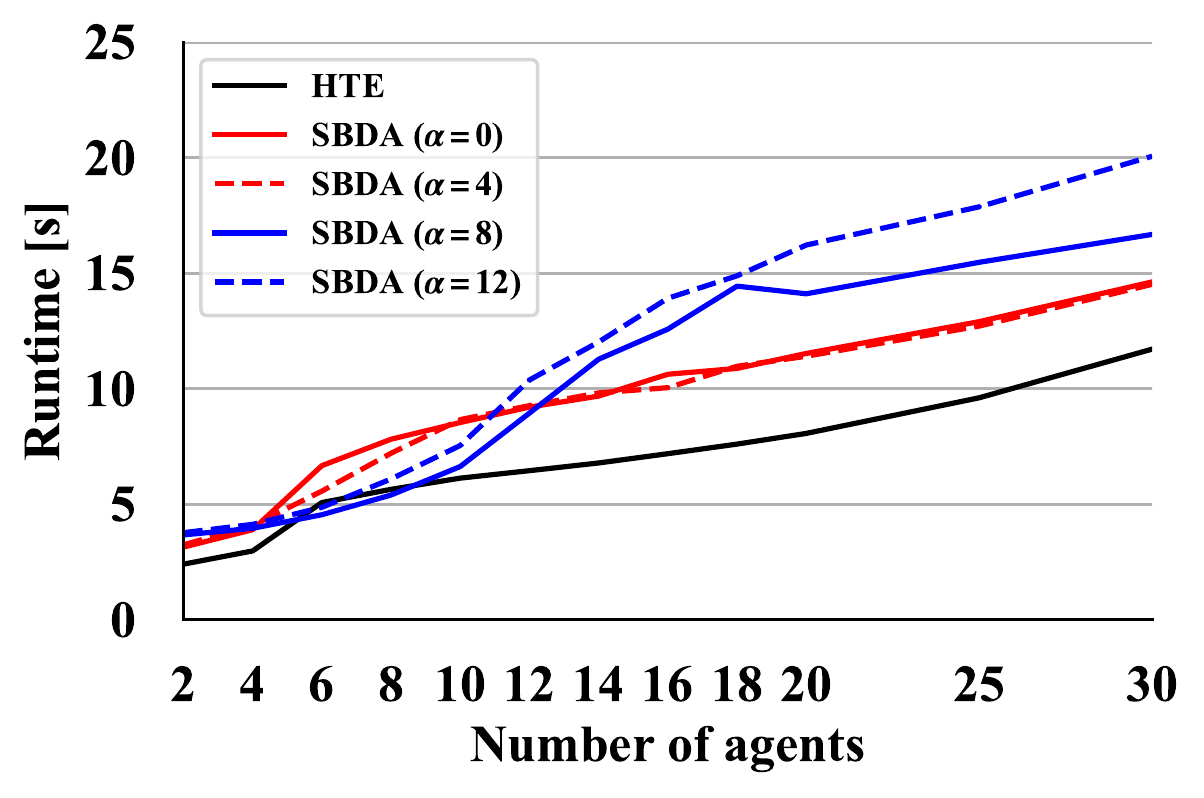}
    \subcaption{Runtime}\label{subfig:exp1-env1-runtime}
  \end{minipage}
  \caption{Comparison of SBDA and HTE (Env.~1)}
  \label{fig:exp1-env1}
\end{figure}

\begin{figure}[t]
  \centering
  \begin{minipage}[b]{0.49\hsize}
    \centering
    \includegraphics[width=0.999\hsize]{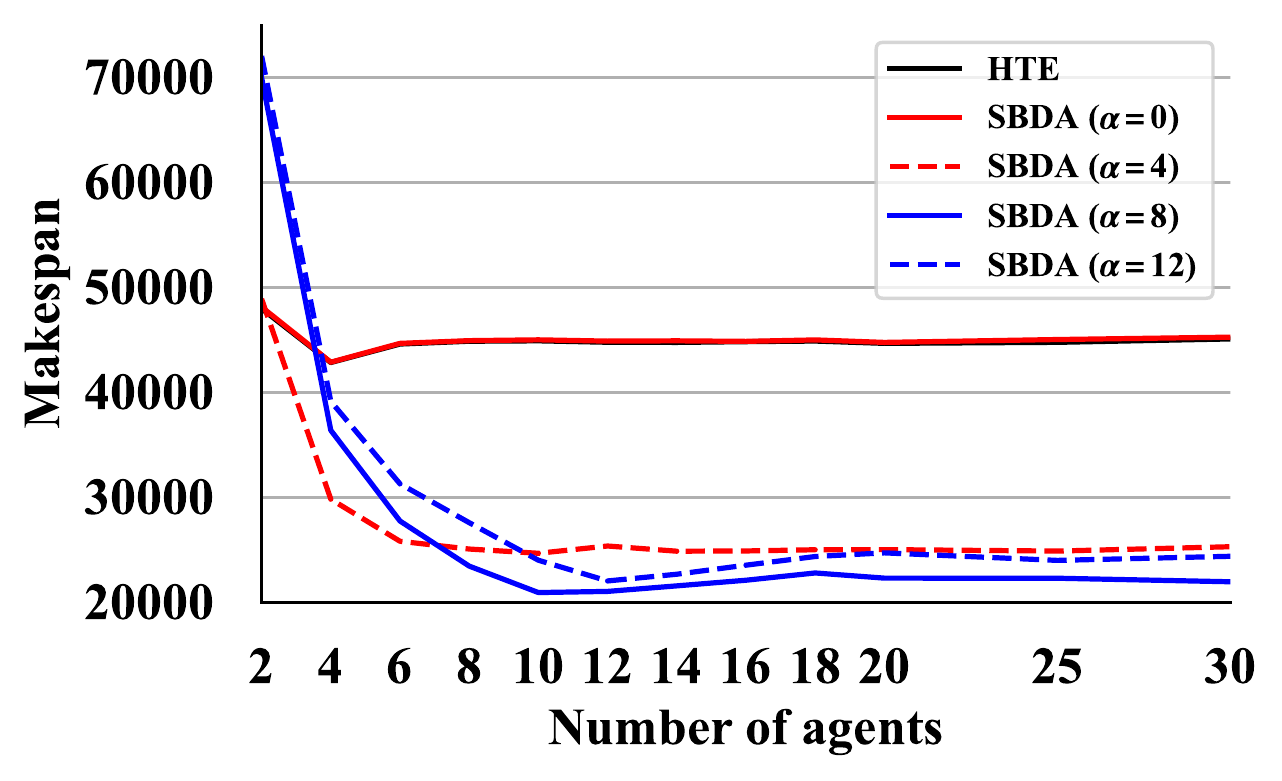}
    \subcaption{Makespan}\label{subfig:exp1-env2-makespan}
  \end{minipage}
  \hfil
  \begin{minipage}[b]{0.49\hsize}
    \centering
    \includegraphics[width=0.999\hsize]{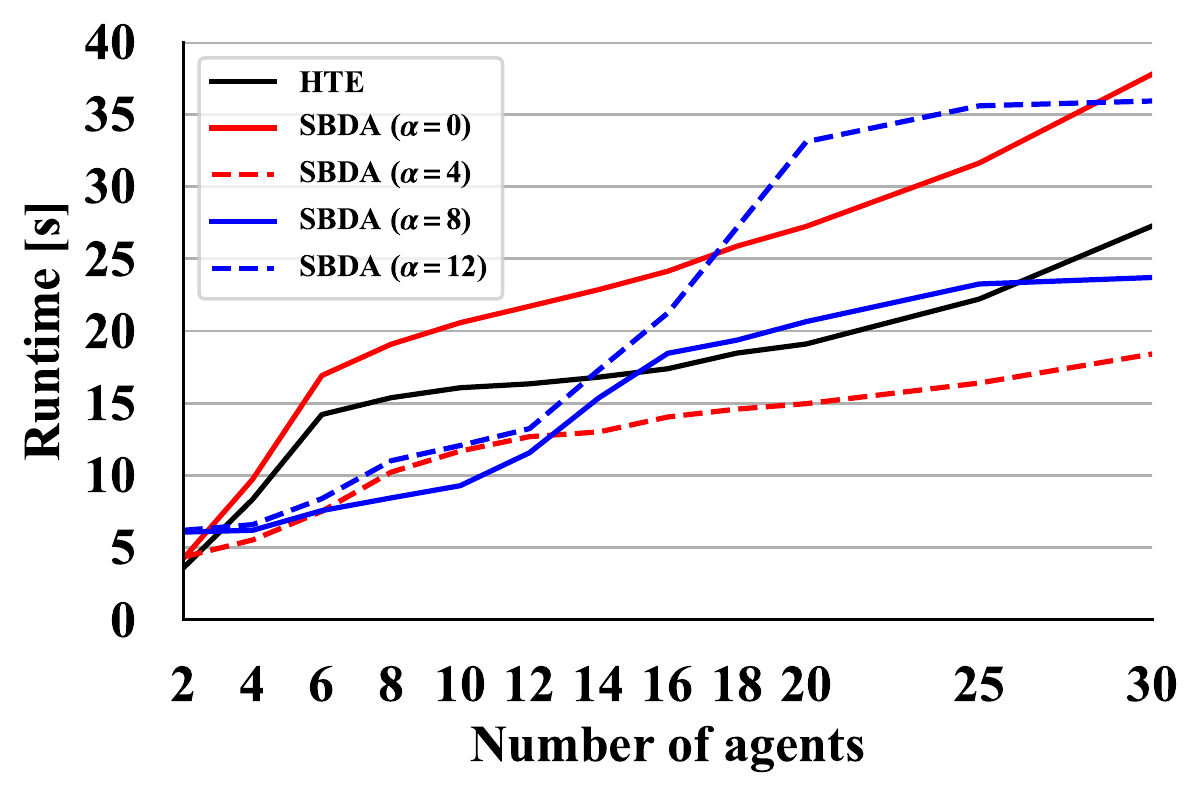}
    \subcaption{Runtime}\label{subfig:exp1-env2-runtime}
  \end{minipage}
  \caption{Comparison of SBDA and HTE (Env.~2)}
  \label{fig:exp1-env2}
\end{figure}

\subsection{Exp.~1: Performance Comparison}
In the first experiment (Exp.~1), we compared the performances of our
proposed SBDA method and the baseline HTE under Envs.~1 and 2. We
plotted the results of Env.~1 in Fig.~\ref{fig:exp1-env1}. We focus
particularly on the results when $\alpha=8$, which exhibited the
highest transportation efficiency (a detailed discussion of the effect
of the parameter $\alpha$ of SBDA is given in
Section~\ref{subsec:exp2}).
\par

Figure~\ref{subfig:exp1-env1-makespan} indicates that SBDA
significantly reduces the makespan compared to the HTE method. In
particular, when the number of agents is $M = 8$, SBDA reduces the
makespan of HTE by approximately 39\%. In HTE, agents cannot
simultaneously execute tasks that have the same endpoint. Hence,
only three or four agents can execute tasks in parallel even if $M$ is
larger because $|V_{tsk}| = 6$ under Env.~1. By contrast, because
multiple agents can simultaneously execute tasks whose loading and unloading
nodes overlapping with other executing tasks by using the standby
nodes effectively, SBDA considerably improves the
makespan through the higher parallelism of the task executions. As an
exception, when $M = 2$, the makespan of SBDA is higher than that
of HTE because when $M$ is extremely small compared to the number
of task endpoints $|V_{tsk}|$, agents can easily find tasks whose
endpoints do not overlap and can always execute pickup-and-delivery
tasks simultaneously without considering the standby nodes.
\par

Figure~\ref{subfig:exp1-env1-runtime} shows that the runtime for SBDA slightly increases compared to HTE. The only CPU usage in HTE comes from the task selection and path planning that generate a collision-free path in the less crowded environment. Hence, even if the number of agents increases, the runtime does not significantly increase, and is only used for a confirmation of the task selection. Meanwhile, in SBDA, many agents can select the tasks to execute and generate many paths mainly to endpoints and standby nodes, thus increasing the runtime. However, considering the parallel behavior, the CPU time spent per agent is not large.
\par

The experiment results under Env.~2, which is plotted in Fig.~\ref{fig:exp1-env2}, indicate that SBDA can improve the performance much more than HTE, and the length of the makespan, i.e., the average time required to complete the MAPD instances, is less than half that of HTE. For example, when $M = 10$, SBDA reduces the makespan of HTE by approximately 53\%, which is a larger improvement than that under Env.~1. Because there are only two loading nodes, the number of agents working simultaneously is limited in HTE, and the makespan remains high. For SBDA, instead of many agents moving toward a small number of loading locations without any control, the agents can automatically wait at the standby nodes close to the loading nodes and move there in turn, which can increase the number of agents working at the same time. Furthermore, because Cond.~1 prevents too many agents from entering the work area, and Cond.~2 prevents the agents from ignoring the agents waiting at the standby nodes and cutting in to move toward the endpoint first, such conditions prevent excessive congestion, and thus contribute to the high efficiency.
\par

\begin{figure*}[t]
  \centering
  \begin{minipage}[b]{0.27\hsize}
    \centering
    \includegraphics[width=0.98\hsize]{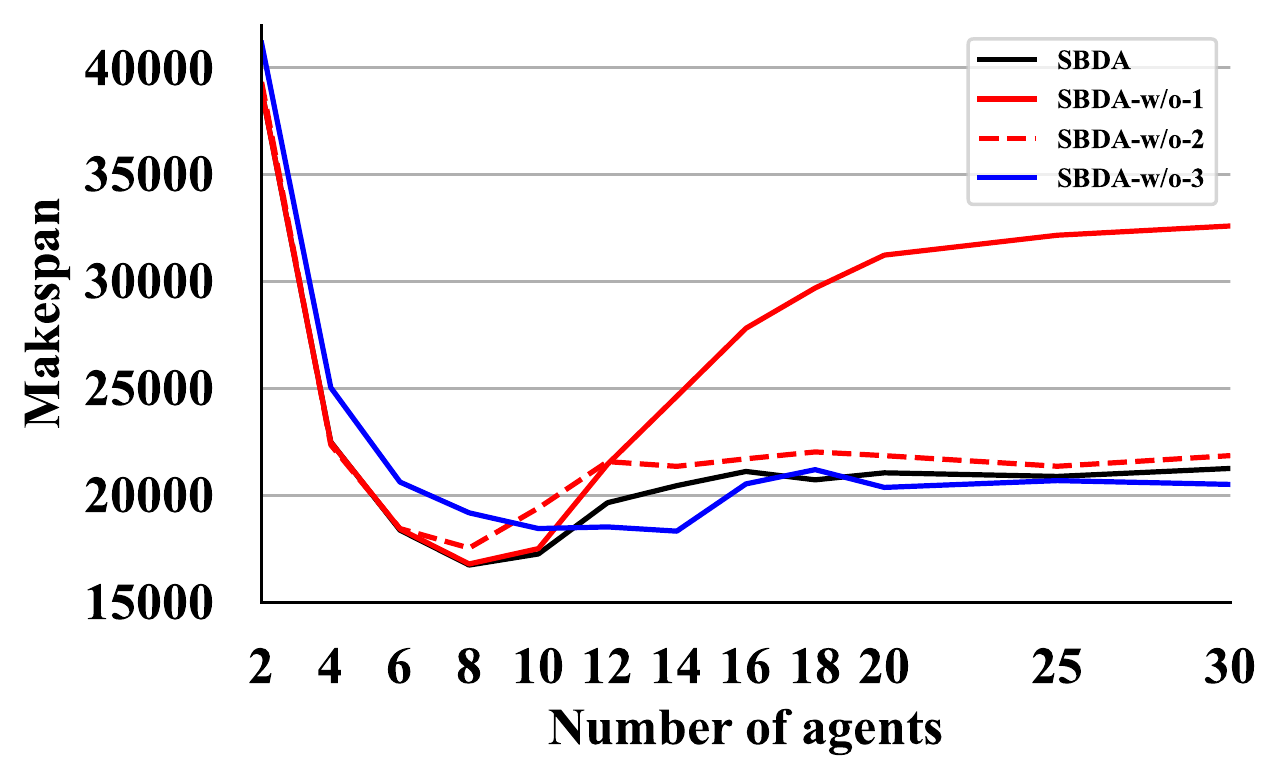}
    \subcaption{Makespan}\label{subfig:exp2-env1-makespan}
  \end{minipage}
  \hfil
  \begin{minipage}[b]{0.27\hsize}
    \centering
    \includegraphics[width=0.92\hsize]{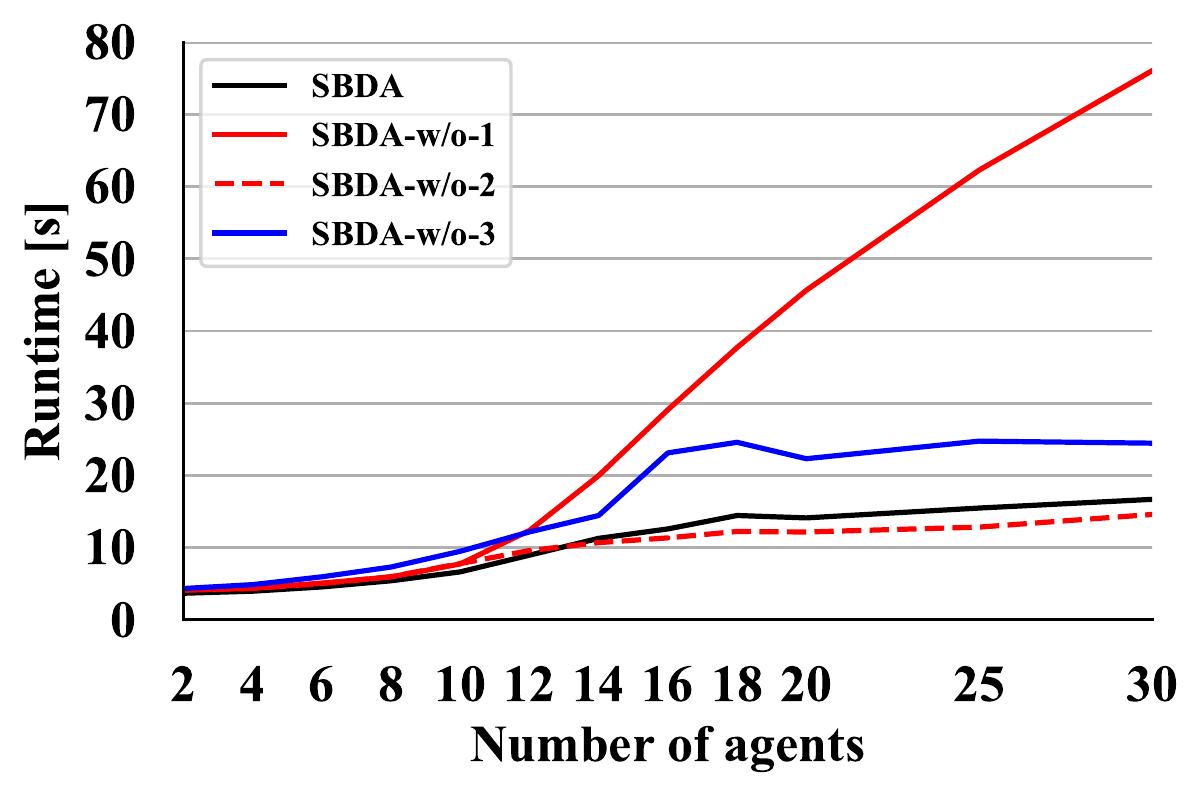}
    \subcaption{Runtime}\label{subfig:exp2-env1-runtime}
  \end{minipage}
  \hfil
  \begin{minipage}[b]{0.27\hsize}
    \centering
    \includegraphics[width=0.98\hsize]{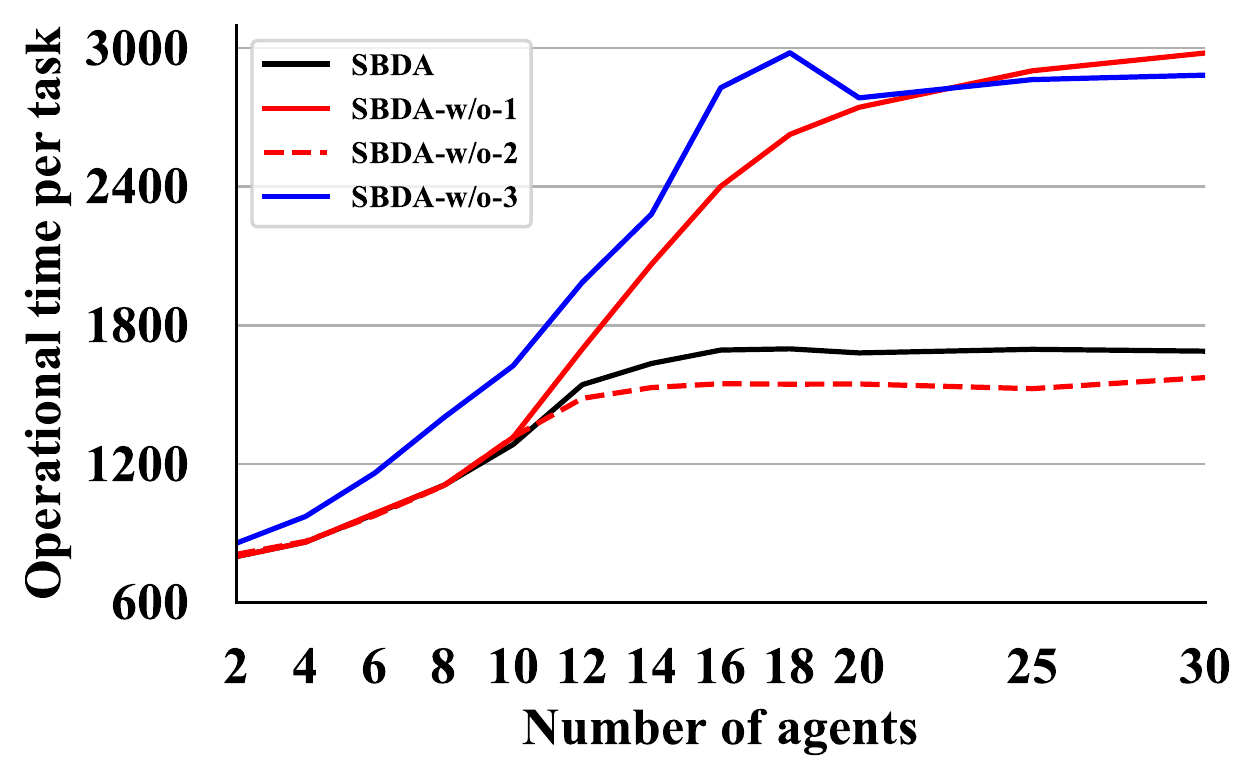}
    \subcaption{Operational time}\label{subfig:exp2-env1-operational-time}
  \end{minipage}
  \caption{Ablation study with task selection conditions (Env.~1)}
  \label{fig:exp2-env1}
\end{figure*}

\begin{figure*}[t]
  \centering
  \begin{minipage}[b]{0.27\hsize}
    \centering
    \includegraphics[width=0.98\hsize]{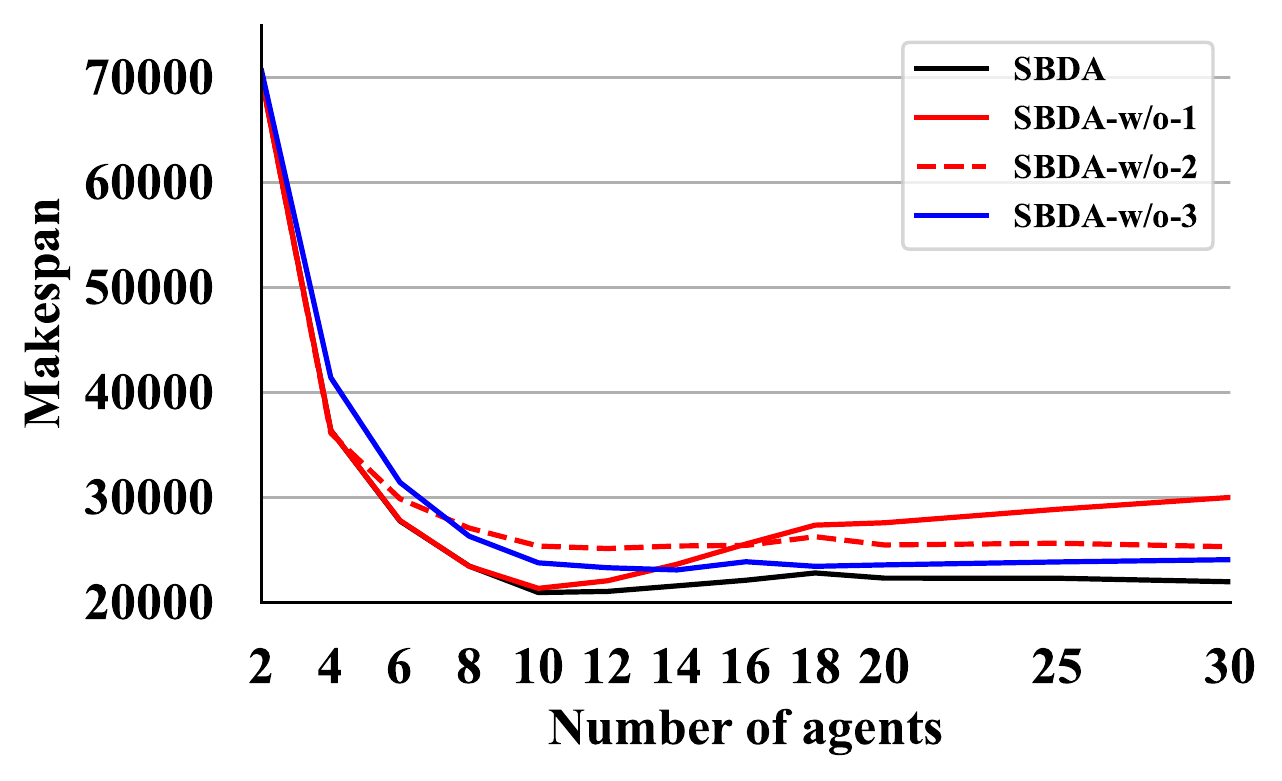}
    \subcaption{Makespan}\label{subfig:exp2-env2-makespan}
  \end{minipage}
  \hfil
  \begin{minipage}[b]{0.27\hsize}
    \centering
    \includegraphics[width=0.95\hsize]{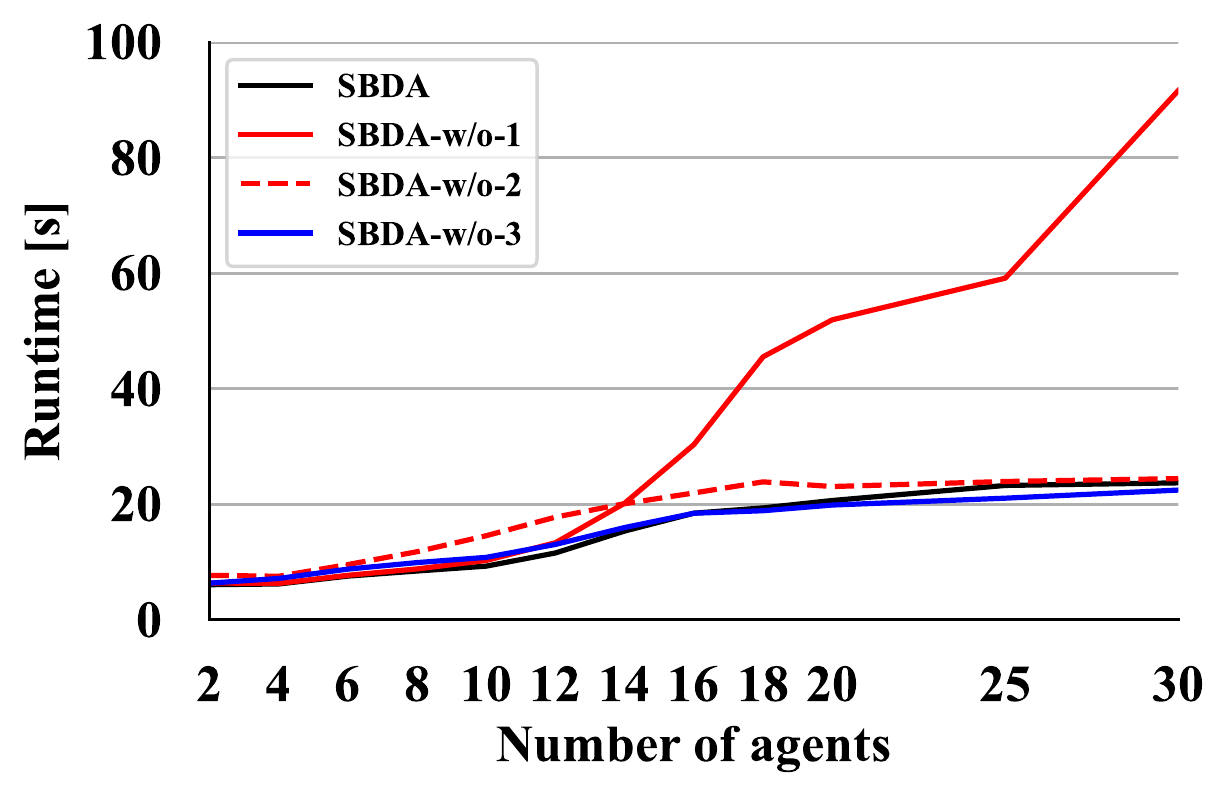}
    \subcaption{Runtime}\label{subfig:exp2-env2-runtime}
  \end{minipage}
  \hfil
  \begin{minipage}[b]{0.27\hsize}
    \centering
    \includegraphics[width=0.98\hsize]{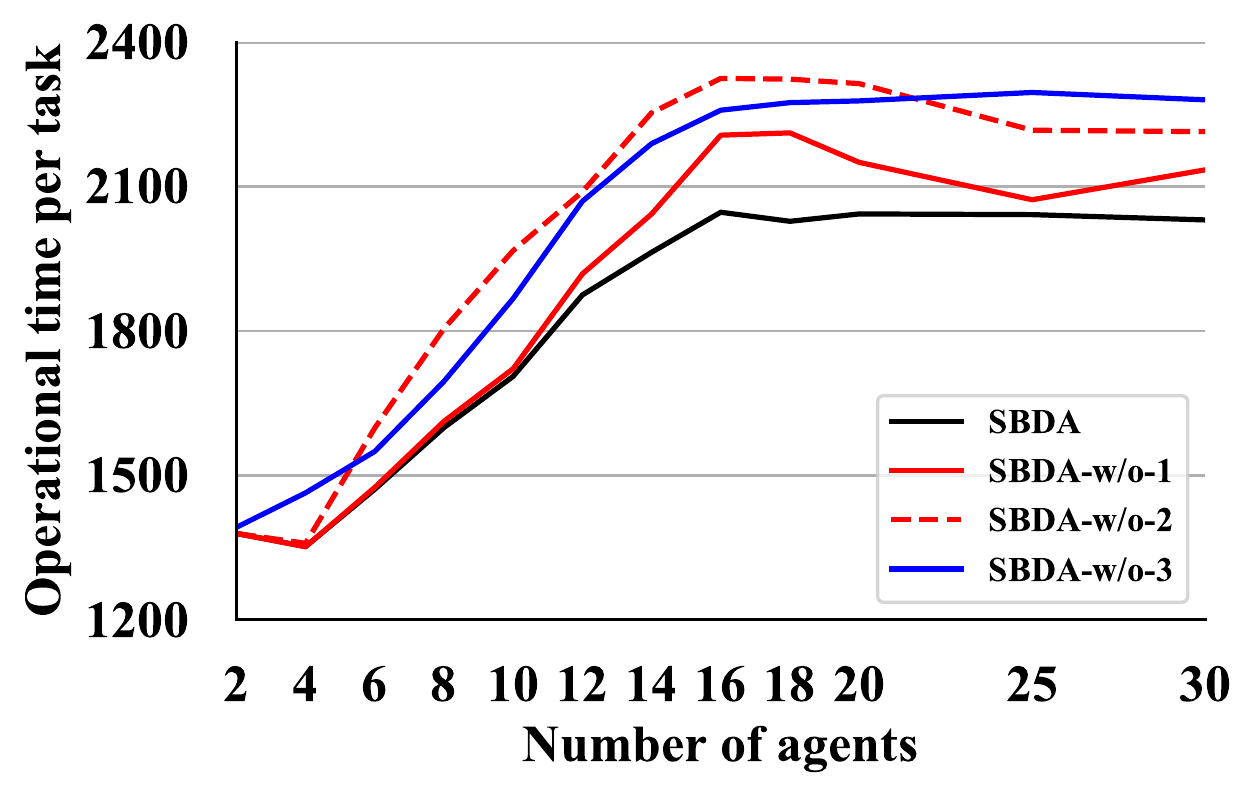}
    \subcaption{Operational time}\label{subfig:exp2-env2-operational-time}
  \end{minipage}
  \caption{Ablation study with task selection conditions (Env.~2)}
  \label{fig:exp2-env2}
\end{figure*}

\subsection{Exp.~2: Features of SBDA}
\label{subsec:exp2}
We conducted a number of experiments to confirm the impact of
parameter $\alpha$, which limits the distance between standby nodes
and the endpoint $v_\tsk$ under Exp.~1, and then, Cond.~1 through an
ablation study in the second experiment (Exp.~2). We also
report the effect of Cond.~2 in the Appendix.
\par

\subsubsection{Impact of parameter $\alpha$}\label{subsubsec:alpha}
First, to analyze the features of SBDA in Exp.~1, we examined the
impact of the parameter $\alpha$ on the performance. We
conducted the same experiments under Envs.~1 and 2 with values of
$\alpha = 0, 4, 8,$ and $12$. The results are also shown in
Figs.~\ref{fig:exp1-env1} and \ref{fig:exp1-env2}. From
Figs.~\ref{subfig:exp1-env1-makespan} and
\ref{subfig:exp1-env2-makespan}, the makespan of both environments
decreases as the value of $\alpha$ increases from zero to 8, and
conversely, increases when $\alpha = 12$. The larger the value of
$\alpha$, the larger the number of associated potential standby nodes
for each task endpoint. This means that (2) and (3) in Cond.~1 of the
task selection (Section~\ref{sec:TaskSelection}) are more easily
satisfied. Hence, more agents left their parking nodes as long as the
environment was not too crowded (Cond.~1 (1)), resulting in an
increase in the parallelism during the task executions and a decrease
of the makespan.
\par

However, when $\alpha$ was too large, such as $\alpha = 12$, the makespan worsened. There seem to be three possible reasons for this. First, too much parallelism in a task execution causes congestion in the environment, leading to larger waiting times and longer detours. In SBDA, agents need to move toward their destinations without passing through the standby nodes. Although the connectivity of the graph $G_t$ is maintained because of the property of the standby nodes, if more nodes in $G_t$ are eliminated, the paths to the endpoints will be lengthened and the agents will be forced to make significant detours, resulting in long wait times at the standby nodes. As the second reason, the agents have to remain in standby nodes far from their task endpoints. Hence, when an agent starts toward an endpoint, the endpoint is exclusively held for a slightly longer period of time until the agent arrives, resulting in a chain reaction of reduced transportation efficiency. Finally, because more agents can fulfill their tasks, local congestion is likely to occur, particularly biasing the destinations.
\par

We can also see from Figs.~\ref{subfig:exp1-env1-makespan} and \ref{subfig:exp1-env2-makespan} that the number of agents for achieving the highest efficiency depends on the value of $\alpha$ as well as the structure of the environment. However, it seems that in all cases, except for $\alpha = 0$ in Env.~2, SBDA exhibited a better performance than HTE for $M\geq 4$. Note that under Env.~2, the performance was the highest when $M = 10$ - $12$, which is slightly larger than that of Env.~1, and did not decrease significantly when $M > 12$, unlike the performance under Env.~1 when $M > 8$. This is because Env.~2 has only two pickup locations, which are thus more prone to localized congestion than Env.~1. Nevertheless, as the number of agents increases, Cond.~1 (2) of task selection can steer agents to tasks whose pickup nodes are less crowded and ultimately eliminate the congestion bias, thus reducing the makespan. In particular, this effect appeared more strongly under Env.~2 than under Env.~1, and thus, Cond.~1 (2) can prevent the agents from entering too much of the work area, averting a loss in efficiency.
\par

If we look at Figs.~\ref{subfig:exp1-env1-runtime} and \ref{subfig:exp1-env2-runtime}, the runtime required for planning is quite large when $\alpha = 0$ and $M$ is between 6 and 10 under Env.~1 and when $M\geq 4$ under Env.~2, although the number of agents moving in the work area seemed small. This may be because agents in the parking nodes frequently probe for a possible task selection, whereas agents with tasks do not do so.
\par

\subsubsection{Ablation study}
Focusing on Cond.~1 of the task selection (Section~\ref{sec:TaskSelection}), we compared SBDA with SBDA without Cond.~1 (1) (SBDA-w/o-1), SBDA without Cond.~1 (2) (SBDA-w/o-2), and SBDA without Cond.~1 (3) (SBDA-w/o-3) as an ablation study to verify the effect of each condition. Figs.~\ref{fig:exp2-env1} and \ref{fig:exp2-env2} show the comparative results of SBDA and ablation methods under Env.~1 and Env.~2, respectively. First, from Fig.~\ref{subfig:exp2-env1-makespan}, the makespan of SBDA-w/o-1 was almost the same as that of SBDA under $M\leq 10$, whereas the makespan increased as $M$ increased after $M = 12$, particularly under Env.~1. This indicates that Cond.~1 (1) contributed significantly to preventing many agents from entering the work area and crowding it. Under Env.~2, there are only two pickup nodes, and Cond.~1 (2) also prevents agents from entering the crowded work area; thus, the makespan did not increase significantly. This analysis also supports the fact that the runtime of SBDA-w/o-1 rapidly increases (Figs.~\ref{subfig:exp2-env1-runtime} and \ref{subfig:exp2-env2-runtime}). Second, the minimum values of the makespan of SBDA-w/o-2 and SBDA-w/o-3 were larger than those of SBDA. If Cond.~1 (2) or (3) was removed, the agent selected a task without considering the status of the load or unload node or the associated potential standby nodes, resulting in a temporary evacuation when heading toward the load or unload node. However, this did not frequently occur in SBDA.
\par

In Figs.~\ref{subfig:exp2-env1-operational-time} and
\ref{subfig:exp2-env2-operational-time}, we also plot the {\em
  operational time per task} or simply the {\em operational time},
which is the average time required to complete one task. In general,
as the parallelism increases, the environment becomes more crowded and
the operational time increases. However, SBDA-w/o-2 has the smallest
value under Env.~1
(Fig.~\ref{subfig:exp2-env1-operational-time}). Because Cond.~1 (2)
was ineffective, many agents evacuated the free potential standby
nodes and $CL$ was likely to be non-empty. As the results indicate,
Cond.~1 (1) limited the number of agents entering the work area,
reducing the parallelism. However, in Env.~2, the operational time of
SBDA-w/o-2 increased because the agents had to stay at the free
potential standby nodes much longer because of the limited number of
loading nodes. Moreover, this also occurred because, by removing
Cond.~1 (2), the agents always selected the tasks with the closest
loading nodes without considering congestion; in particular, the
agents in the parking nodes were biased toward to the upper-left
loading node under Env.~2. In other cases, i.e., in SBDA-w/o-1 and
SBDA-w/o-3, the operational time was large because too many agents entered the work area.
\par

Finally, Figs.~\ref{subfig:exp2-env1-operational-time} and
\ref{subfig:exp2-env2-operational-time} show that the operational time
of SBDA-w/o-3 tended to be large under both environments. Through the
removal of Cond.~1 (3), the agents were likely to select tasks without
considering the status of the unload nodes. Hence, they could not move
to the unload nodes or their associated potential standby nodes after
the load was completed and were repeatedly forced to temporarily
remain at free potential standby nodes or parking nodes with carrying the materials. Such interruptions in task execution lengthened the operation time.
\par


\section{Conclusion}
We presented a deadlock avoidance method, called SBDA, for the MAPD problem to improve the transportation efficiency even in a maze-like environment. The central idea of SBDA is the use of {\em standby nodes} in which, even if an agent stays there, the connectivity of the environment remains. Furthermore, standby nodes can be identified with a low computational cost using the articulation-point-finding algorithm used in graph theory and are thereby found in real-time. SBDA guarantees completeness for well-formed MAPD instances by effectively using standby nodes that are guaranteed to wait for any finite amount of time. We evaluated the proposed method in comparison with a well-known conventional method, HTE, in restricted maze-like environments based on our envisioned applications. Our experiment results demonstrated that SBDA considerably outperforms the conventional method. We also analyzed the features of our proposed method and observed that the parallelism of the task execution can be controlled by changing the value of $\alpha$, the parameter used for selecting standby nodes, and that the transportation efficiency can be significantly improved by appropriately setting the value of $\alpha$.
\par

To improve the flexibility and usability of SBDA for real-world applications, we plan to study a method for determining the appropriate value of $\alpha$ from a graph structure.



\begin{acks}
This work was partly supported by JSPS KAKENHI Grant Numbers 20H04245 and 17KT0044.
\end{acks}



\bibliographystyle{ACM-Reference-Format}
\bibliography{reference.bib}


\clearpage

\begin{figure}
  \centering
  \begin{minipage}[b]{0.68\hsize}
    \centering
    \includegraphics[width=0.95\hsize]{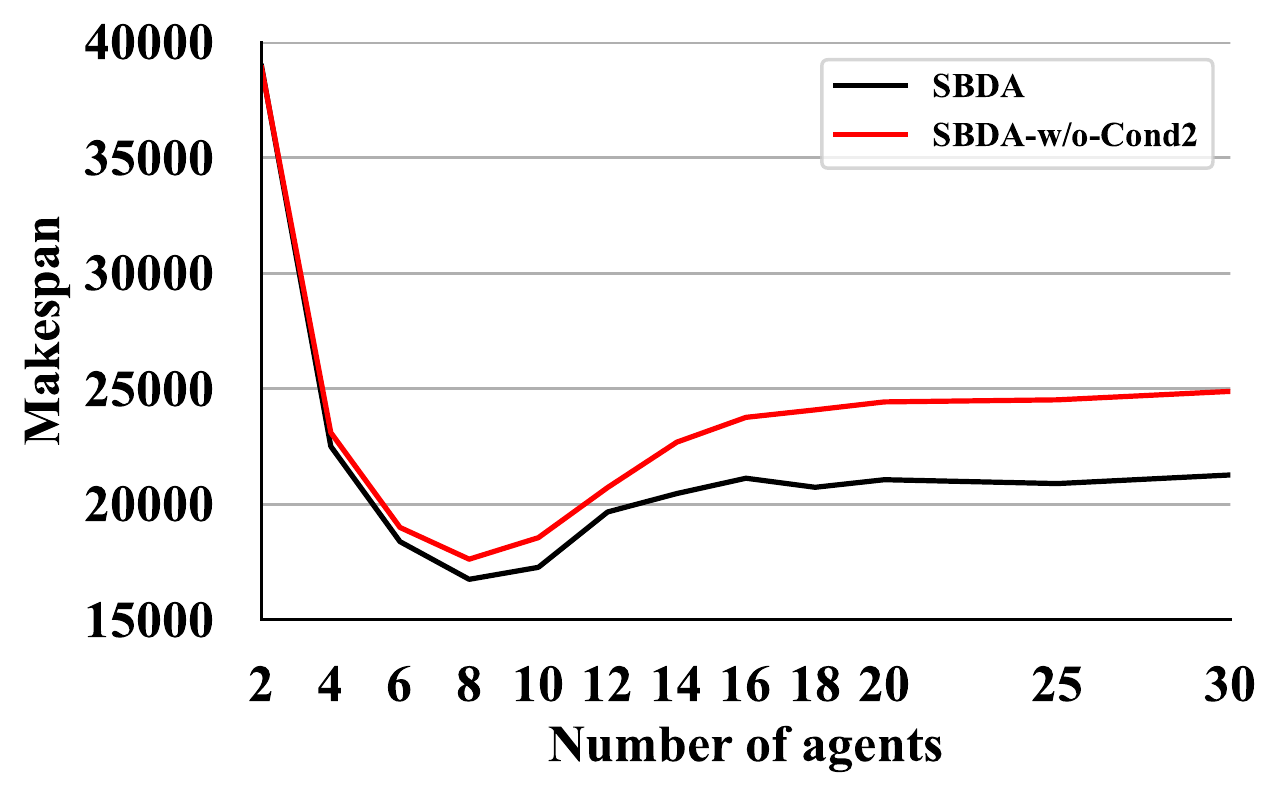}
    \subcaption{Env.~1}\label{subfig:env1-makespan}
  \end{minipage}
  \hfil
  \begin{minipage}[b]{0.68\hsize}
    \centering
    \includegraphics[width=0.95\hsize]{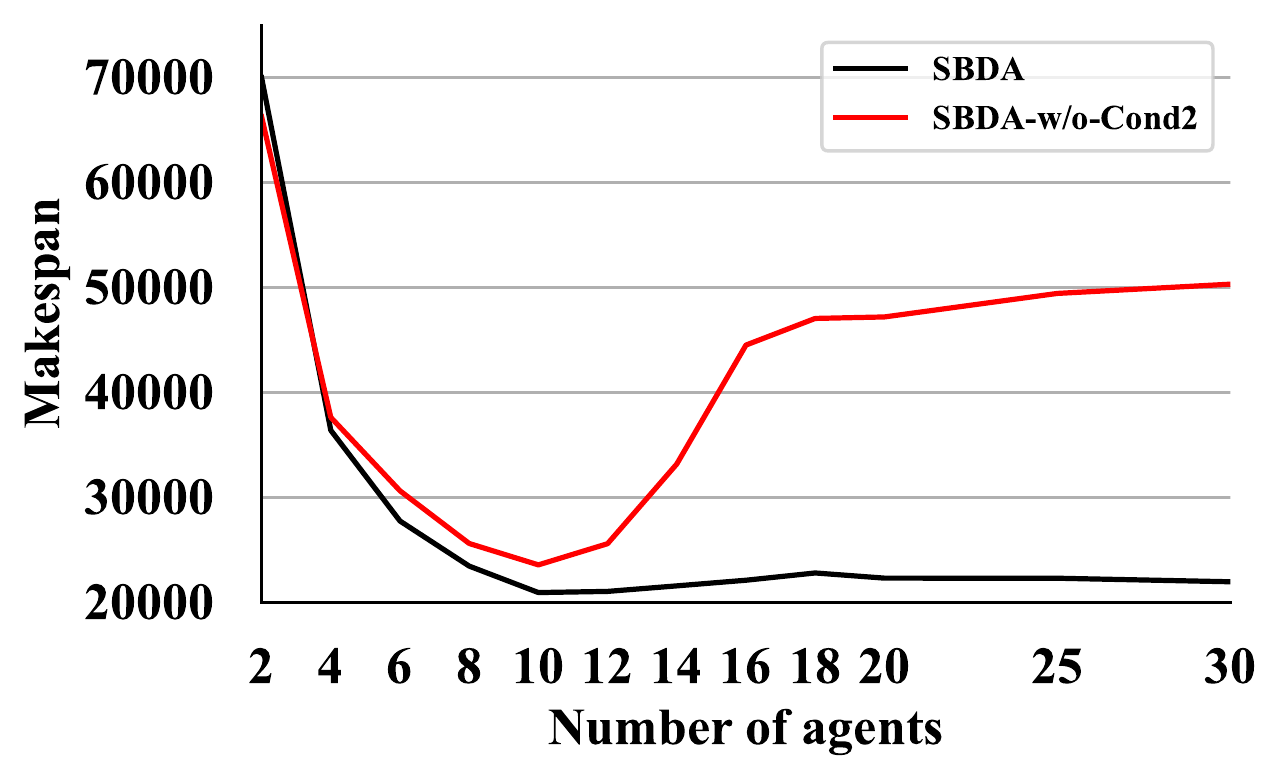}
    \subcaption{Env.~2}\label{subfig:env2-makespan}
  \end{minipage}
  \caption{Ablation study with Cond.~2 (Makespan)}
  \label{fig:makespan}
\end{figure}

\begin{figure}
  \centering
  \begin{minipage}[b]{0.68\hsize}
    \centering
    \includegraphics[width=0.95\hsize]{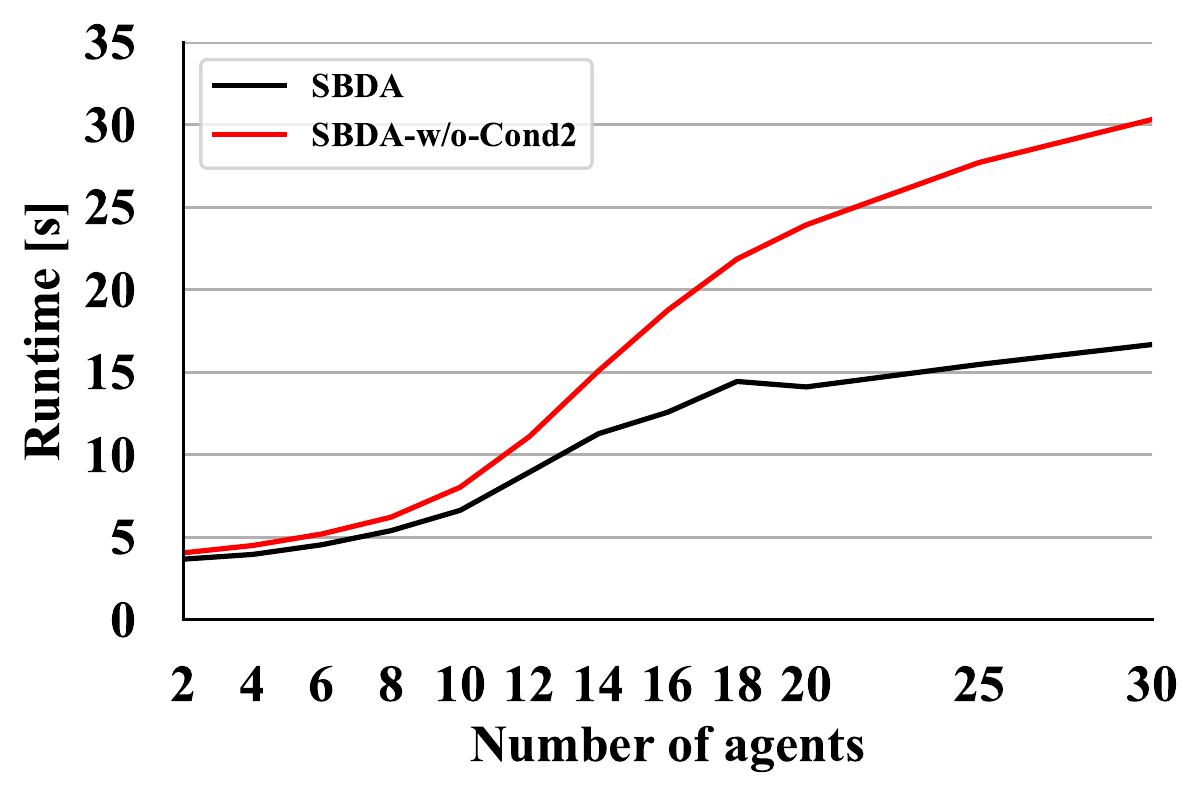}
    \subcaption{Env.~1}\label{subfig:env1-runtime}
  \end{minipage}
  \hfil
  \begin{minipage}[b]{0.68\hsize}
    \centering
    \includegraphics[width=0.95\hsize]{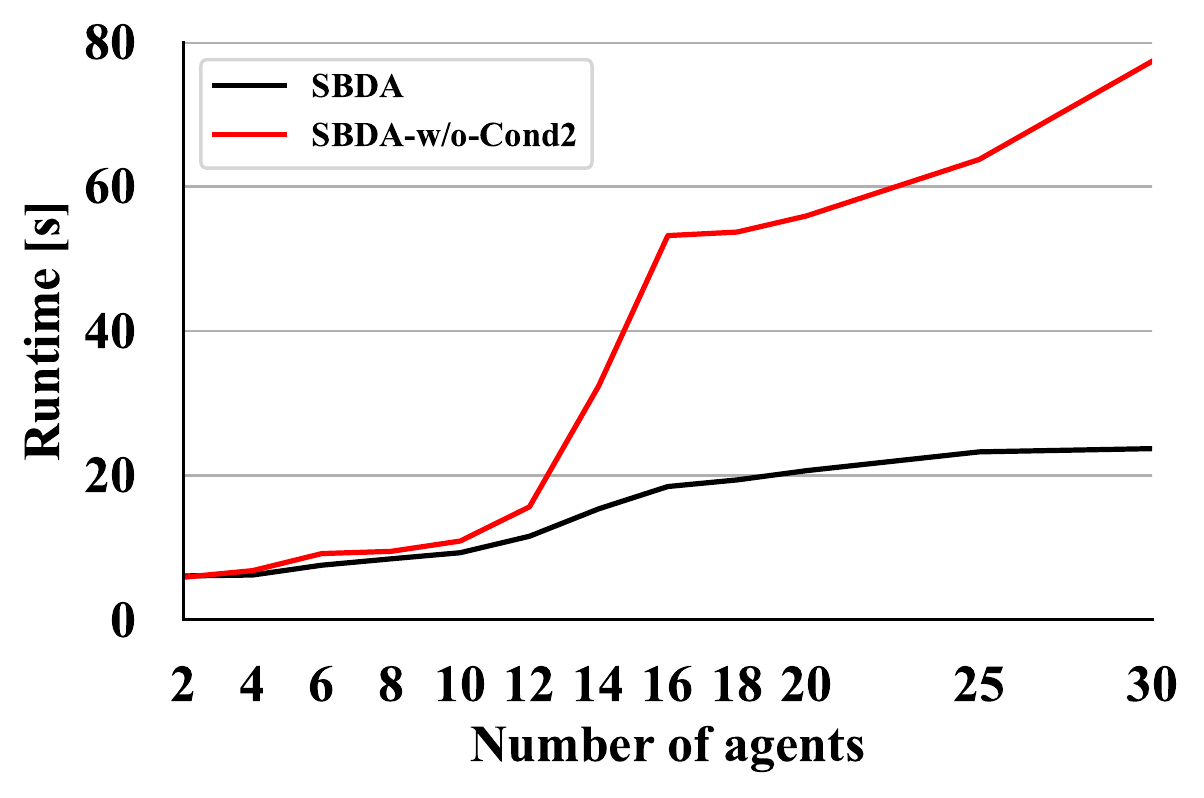}
    \subcaption{Env.~2}\label{subfig:env2-runtime}
  \end{minipage}
  \caption{Ablation study with Cond.~2 (Runtime)}
  \label{fig:runtime}
\end{figure}

\begin{figure}
  \centering
  \begin{minipage}[b]{0.68\hsize}
    \centering
    \includegraphics[width=0.95\hsize]{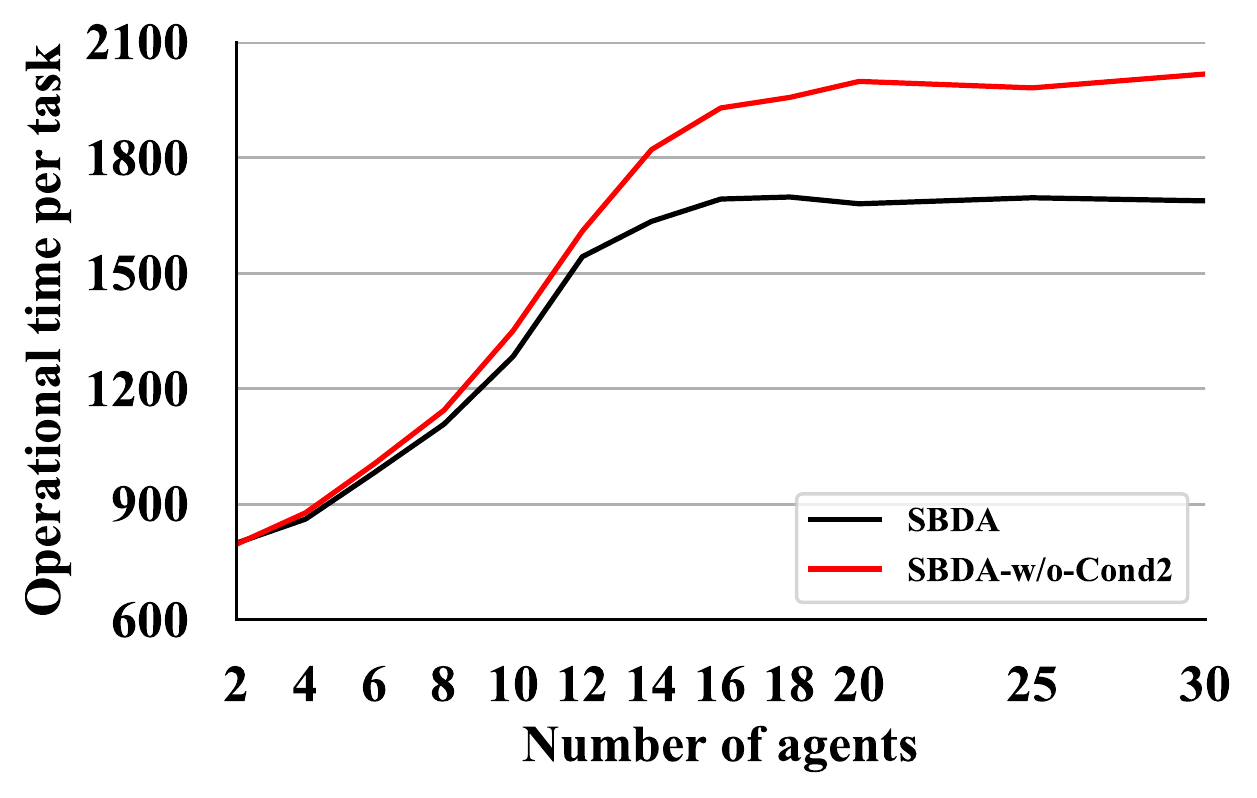}
    \subcaption{Env.~1}\label{subfig:env1-operational-time}
  \end{minipage}
  \hfil
  \begin{minipage}[b]{0.68\hsize}
    \centering
    \includegraphics[width=0.95\hsize]{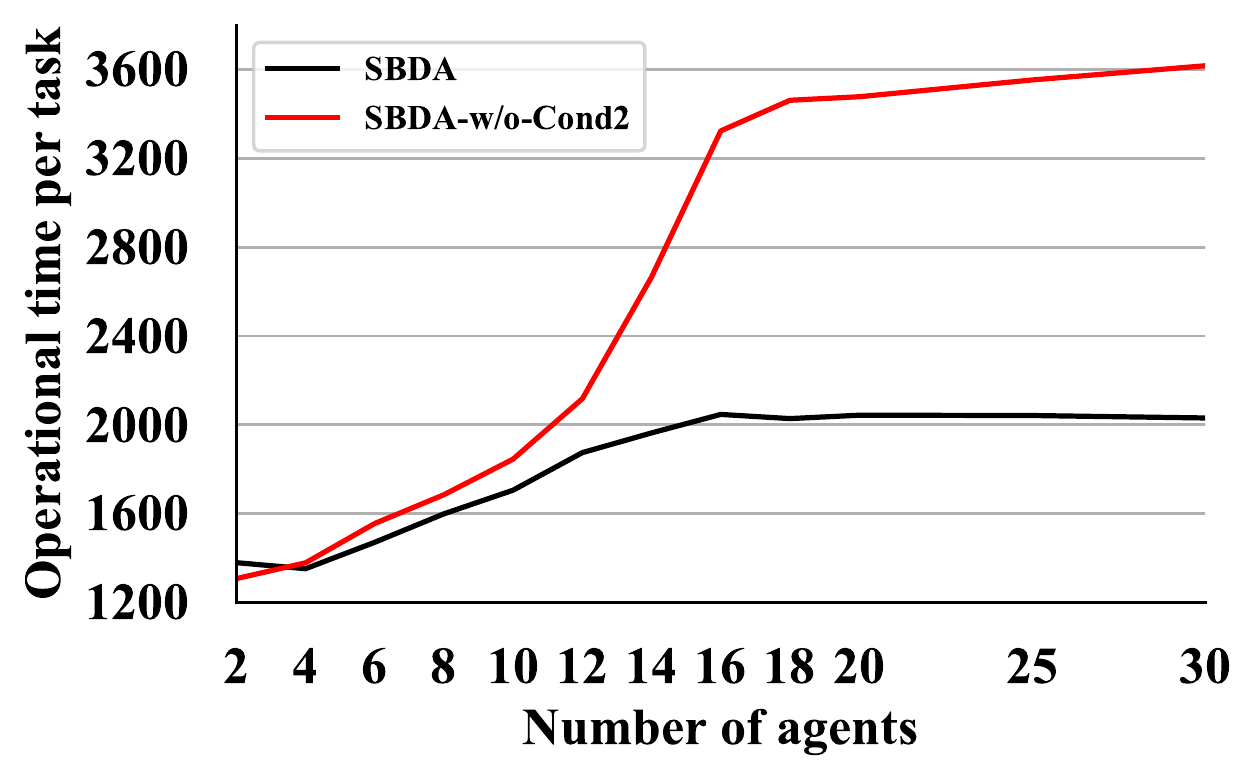}
    \subcaption{Env.~2}\label{subfig:env2-operational-time}
  \end{minipage}
  \caption{Ablation study with Cond.~2 (Operational time)}
  \label{fig:operational-time}
\end{figure}

\section*{Appendix}
\subsection*{Additional Ablation Study}
We also conducted experiments to verify the effect of Cond.~2 for the destination decision in the ablation study as Exp.~2. Figs.~\ref{fig:makespan}, \ref{fig:runtime} and \ref{fig:operational-time} are plotted the results of SBDA and ablation method of SBDA without Cond.~2 (SBDA-w/o-Cond2) in Envs.~1 and 2.
\par

These figures indicate that the makespan, the runtime, and the operational time of SBDA-w/o-Cond2 significantly increased as $M$ increased after $M = 12$ in both Envs.~1 and 2, especially in Env.~2. This indicates that Cond.~2 greatly contributed to preventing agents from ignoring agents waiting in standby nodes and cutting into the queue to endpoints. In fact, if Cond.~2 is removed, depending on whether the endpoint is open when the agent is in a distant standby node, task endpoint, or parking node, it may hold the endpoint, and other agents on standby nodes near the endpoint will be forced to wait for an additional long time, especially in Env.~2 because it is wider than Env.~1 and has only two load nodes. Such long waiting cascades other delays in movements of other agents, and makespan worsened significantly.

\end{document}